\documentclass[a4paper]{article}
\usepackage{amsmath,amsfonts,amsthm}
\usepackage{amssymb,latexsym}

\theoremstyle{plain}
\newtheorem{theorem}{Theorem}[section]
\newtheorem{lemma}[theorem]{Lemma}
\newtheorem{proposition}[theorem]{Proposition}
\newtheorem{corollary}[theorem]{Corollary}

\theoremstyle{definition}
\newtheorem{mydefn}[theorem]{Definition}
\newtheorem{example}[theorem]{Example}

\newtheorem{rmk}[theorem]{Remark}

\usepackage{mathrsfs}
\usepackage{amsfonts}
\usepackage{graphicx}
\usepackage{slashed}

\begin{document}

\author{Yan Zhang$^1$ \;
Zhaohui Zhu$^1$ \footnote{Corresponding author. Email: zhaohui@nuaa.edu.cn, zhaohui.nuaa@gmail.com } \;
Jinjin Zhang$^2$  \;
Yong Zhou$^1$ \\
1 College of Computer Science \\Nanjing University of Aeronautics and Astronautics\\
2 College of Information Science \\ Nanjing Audit University}
%\ead{zhyi812@163.com}

%\author{Zhaohui Zhu}
%\email{zhaohui@nuaa.com}
%
%\author[nanshen]{Jinjin Zhang}
%%\ead{zhangjinjin@163.com}
%
%\author[nuaa]{Yong Zhou}
%%\ead{zhouyong@163.com}
%
%\cortext[cor]{Corresponding author. Email: zhaohui@nuaa.edu.cn, zhaohui.nuaa@gmail.com (Zhaohui Zhu).}
%
%\address[nuaa]{College of Computer Science, Nanjing University of Aeronautics and Astronautics, Nanjing, P.R. China, 210016}
%\address[nanshen]{College of Information Science, Nanjing Audit University, \\Nanjing, P.R. China, 211815}

\title{Axiomatizing L\"{u}ttgen \& Vogler's ready simulation for finite processes in $\text{CLL}_R$ \footnote{This work received financial support of the National Natural Science of China (No. 60973045) and Fok Ying-Tung Education Foundation.}}

\date{\today}

%\begin{keyword}
%    process calculus \sep ready simulation \sep logic labelled transition system \sep axiomatization \sep $\text{CLL}_R$
%\end{keyword}

\maketitle
%
%\begin{thanks}
%  hello,don't
%\end{thanks}
\begin{abstract}

In the framework of logic labelled transition system, a variant of weak ready simulation has been presented by L{\"u}ttgen and Vogler.
It has been shown that such behavioural preorder is the largest precongruence w.r.t parallel and conjunction composition satisfying desired properties.
This paper offers a ground-complete axiomatization for this precongruence over processes containing no recursion in the calculus $\text{CLL}_R$.
Compared with usual inference system for process calculus, in addition to axioms about process operators, such system contains a number of axioms to characterize the interaction between process operators and logical operators.

\textbf{Keywords:} process calculus, weak ready simulation, logic labelled transition system, axiomatization, $\text{CLL}_R$
\end{abstract}

\section{Introduction}

    It is well-known that process algebra and temporal logic take different standpoint for looking at specifications and verifications of reactive and concurrent systems, and offer complementary advantages \cite{Peled01}.
    To take advantage of these two paradigms when designing systems, a few theories for heterogeneous specifications have been proposed, e.g., \cite{Cleaveland00, Cleaveland02, Graf86,Kurshan94,Luttgen07,Luttgen10,Luttgen11,Olderog}.
    Among them, L{\"u}ttgen and Vogler propose the notion of logic labelled transition system (Logic LTS or LLTS for short), which combines operational and logical styles of specification in one unified framework \cite{Luttgen07,Luttgen10,Luttgen11}.
    %Roughly speaking, a LLTS is a LTS with a predicate on states.
%    A few of constructors over LLTSs are introduced, which include operational constructors, such as CSP-style parallel composition and hiding, and logic constructors conjunction and disjunction.
    In particular, a variant of weak ready simulation has been presented in \cite{Luttgen10}, which is adopted to capture refinement relation between processes in the presence of logical operators.
    It has been shown that such simulation is the largest precongruence w.r.t parallel and conjunction satisfying desired properties \cite{Luttgen10}.
%    This framework allows ones to freely mix operational and logic operators.
    Moreover, in addition to usual process operators (e.g., CSP-style parallel composition, hiding, etc) and logic operators (disjunction and conjunction), some standard temporal logic operators, such as \textquotedblleft always\textquotedblright and \textquotedblleft unless\textquotedblright, are also integrated into this framework \cite{Luttgen11}.
    In a word, L\"{u}ttgen and Vogler offer a framework which allows ones to freely mix operational and logic operators when designing systems.

   L\"{u}ttgen and Vogler's approach is entirely semantic, and doesn't provide any kind of syntactic calculus.
   Recently, the first three authors of this paper explore recursive operations over LLTS in a pure process-algebraic style.
   A LLTS-oriented process calculus $\text{CLL}_R$ is presented, and the uniqueness of solutions of equations in $\text{CLL}_R$ is established under a certain circumstance \cite{Zhang14}.

It is one of important topics in concurrency theory that giving axiomatization for behaviour relations.
For example, Milner gives an axiomatization for observational congruence in CCS \cite{Milner89b};
Baeten and Bravetti extend Milner's this work and provide an axiomatization over $\text{TCP+REC}_f$ \cite{Baeten08}, where $\text{TCP+REC}_f$ is a fragment of $\text{TCP+REC}$ which is a generic process language that embodies features of the classical process algebras CCS, CSP and ACP;
Lin offers complete inference systems for late and early weak bisimulation equivalences for processes without involving recursion in $\pi$-calculus \cite{Lin};
  Aceto et al. explore the axiomatization of weak simulation semantics systematically over BCCSP (without recursion) \cite{Aceto13}.
 % Lohrey et al. axiomatise divergence (caused by recursion) systematically \cite{Lohrey}.
% % They give two soundness and ground-complete proof systems for weak ready simulation congruence.
%  %one has conditional axiom and the other does not.\\
%As we know, always recursion is not easy to tackle axiomatically.
  Although L\"{u}ttgen and Vogler's original paper \cite{Luttgen10} mentions some sound laws, a complete set of axioms seems out of reach.
   As the main contribution of this paper we intend to provide a proof system for L\"{u}ttgen and Vogler's weak ready simulation  over  $\text{CLL}_R$-processes with finite behaviour, and demonstrate its soundness and ground-completeness.

    The rest of this paper is organized as follows.
    The notion of Logic LTS and the calculus $\text{CLL}_R$ are recalled in the next section.
    The inference system is presented in Section~3, along with the soundness proof.
    Section~4 demonstrates that the inference system is ground-complete for processes with finite behaviour.
    The paper is concluded with Section~5, where a brief discussion is given.

\section{Preliminaries}

The purpose of this section is to fix our notation and terminology, and to introduce some concepts that underlie our work in all other parts of the paper.

\subsection{Logic LTS and ready simulation}

%In this subsection, we introduce some useful notations and recall the definition of Logic LTS and ready simulation \cite{Luttgen11}.

Let $Act$ be the set of visible action names ranged over by $a$, $b$, etc., and let $Act_{\tau}$ denote $Act \cup \{\tau\}$ ranged over by $\alpha$ and $\beta$, where $\tau$ represents invisible actions.
A labelled transition system with predicate is a quadruple $(P,Act_{\tau},\rightarrow,F)$, where $P$ is a set of states, $\rightarrow \subseteq P\times Act_{\tau}\times P$ is the transition relation and $F\subseteq P$.

As usual, we write $p \stackrel{\alpha}{\rightarrow}$ (or, $p \not \stackrel{\alpha}{\rightarrow}$) if $\exists q\in P.p\stackrel{\alpha}{\rightarrow}q$ ($\nexists q\in P.p  \stackrel{\alpha}{\rightarrow}q$, resp.).
The ready set $\{\alpha \in Act_{\tau}|p \stackrel{\alpha}{\rightarrow}\}$ of a given state $p$ is denoted by $\mathcal{I}(p)$.
A state $p$ is stable if $p \not\stackrel{\tau}{\rightarrow}$.
A number of useful decorated transition relations are given:

$p \stackrel{\alpha}{\rightarrow}_F q$ iff $p \stackrel{\alpha}{\rightarrow} q$ and $p,q\notin F$;

$p \stackrel{\epsilon}{\Rightarrow}q$ iff $p (\stackrel{\tau}{\rightarrow})^* q$, where $(\stackrel{\tau}{\rightarrow})^* $ is the transitive and reflexive closure of $\stackrel{\tau}{\rightarrow}$;

$p \stackrel{\alpha}{\Rightarrow}q$ iff $\exists r,s\in P.p \stackrel{\epsilon}{\Rightarrow} r \stackrel{\alpha}{\rightarrow}s \stackrel{\epsilon}{\Rightarrow} q$;

$p \stackrel{\gamma}{\Rightarrow}|q$ iff $p \stackrel{\gamma}{\Rightarrow}q \not\stackrel{\tau}{\rightarrow}$ with $\gamma \in Act_{\tau}\cup \{\epsilon\}$;

$p\stackrel{\epsilon }{\Rightarrow }_Fq$ iff there exists a sequence of $\tau$-transitions from $p$ to $q$ such that all states along this sequence, including $p$ and $q$, are not in $F$; the decorated transition $p \stackrel{\alpha }{\Rightarrow }_Fq$ may be defined similarly;
%$p \stackrel{\epsilon}{\Rightarrow}_F q$ (or, $p \stackrel{\alpha}{\Rightarrow}_F q$) iff  $p \stackrel{\epsilon}{\Rightarrow}q$ ($p \stackrel{\alpha}{\Rightarrow}q$ respectively) and all states along the sequence, including $p$ and $q$, are not in $F$.

$p \stackrel{\gamma}{\Rightarrow}_F|q$ iff $p \stackrel{\gamma}{\Rightarrow}_F q \not\stackrel{\tau}{\rightarrow}$ with $\gamma \in Act_{\tau} \cup \{\epsilon\}$.

 Notice that
the notation $p
\stackrel{\gamma }{\Longrightarrow }\mspace{-8mu}|q$ in \cite{Luttgen10,Luttgen11}
has the same meaning as $p\stackrel{\gamma }{
\Rightarrow }_F|q$ in this paper, while $p\stackrel{\gamma }{\Rightarrow }|q $ in this paper does not involve any requirement on $F$-predicate.

\begin{mydefn}[Logic LTS \cite{Luttgen10}]\label{D:LLTS}
 An LTS $(P,Act_{\tau},\rightarrow,F)$ is an LLTS if, for each $p \in P$,

\noindent\textbf{(LTS1) }$p \in F$ if $\exists\alpha\in \mathcal{I}(p)\forall q\in P(p \stackrel{\alpha}{\rightarrow}q \;\text{implies}\; q\in F)$;

\noindent\textbf{(LTS2)} $p \in F$ if $\nexists q\in P.p \stackrel{\epsilon}{\Rightarrow}_F|q$.

Moreover, an LTS $(P,Act_{\tau},\rightarrow,F)$ is {$\tau$}-pure if, for each $p \in P$, $p\stackrel{\tau}{\rightarrow}$ implies $\nexists a\in Act.\;p\stackrel{a}{\rightarrow}$.
\end{mydefn}

%Any state $p$ in a $\tau$-pure LTS represents either an external or internal choice between its outgoing transitions.
Compared with usual LTSs, one distinguishing feature of LLTS is that it
involves consideration of inconsistencies.
The main motivation
behind such consideration lies in dealing with inconsistencies caused by
conjunctive composition.
In the notion above, the predicate $F$ is used to denote the set of all inconsistent states that represent empty behaviour that cannot be implemented \cite{Luttgen11}.
In the sequel, we shall use the phrase ``{\it %
inconsistency} {\it predicate}'' to refer to $F$.
The condition (LTS1) formalizes the backward propagation of inconsistencies, and (LTS2) captures
the intuition that divergence (i.e., infinite sequences of $\tau $%
-transitions) should be viewed as catastrophic.
For more intuitive ideas and motivation about inconsistency, the reader may refer \cite{Luttgen07,Luttgen10}.

The notion of ready simulation below is adopted to capture the refinement relation in \cite{Luttgen10,Luttgen11}, which is a variant of the usual notion of weak ready simulation \cite{Bloom95,Larsen91}.
It has been proven that such kind of ready simulation is the largest precongruence w.r.t parallel composition and conjunction which satisfies  the desired property that  an inconsistent specification can only be refined by inconsistent ones (see Theorem 21 in \cite{Luttgen10}).
%Such kind of ready simulation cares only stable consistent states.
%and it is fully abstract w.r.t conjunction and parallel composition on LLTS.

\begin{mydefn}[Ready simulation on LLTS \cite{Luttgen10}]\label{D:RS}
Let $(P, Act_{\tau}, \rightarrow , F)$ be a  LLTS.
A relation ${\mathcal R} \subseteq P \times P$ is a stable ready simulation relation, if for any $(p,q) \in {\mathcal R}$ and $a \in Act $\\
\textbf{(RS1)} both $p$ and $q$ are stable;\\
\textbf{(RS2)} $p \notin F$ implies $q \notin F$;\\
\textbf{(RS3)} $p \stackrel{a}{\Rightarrow}_F|p'$ implies $\exists q'.q \stackrel{a}{\Rightarrow}_F|q'\; \textrm{and}\;(p',q') \in {\mathcal R}$;\\
\textbf{(RS4)} $p\notin F$ implies ${\mathcal I}(p)={\mathcal I}(q)$.

\noindent We say that $p$ is stable ready simulated by $q$, in symbols $p \underset{\thicksim}{\sqsubset}_{RS} q$, if there exists a stable ready simulation relation $\mathcal R$ with $(p,q) \in {\mathcal R}$.
 Further, $p$ is ready simulated by $q$, written $p\sqsubseteq_{RS}q$, if
 $\forall p'(p\stackrel{\epsilon}{\Rightarrow}_F| p' \;\text{implies}\; \exists q'(q \stackrel{\epsilon}{\Rightarrow}_F| q'\; \text{and}\;p' \underset{\thicksim}{\sqsubset}_{RS} q'))$.
 The kernels of $\underset{\thicksim}{\sqsubset}_{RS}$ and $\sqsubseteq_{RS}$ are denoted by $\approx_{RS}$ and $=_{RS}$ resp..
 It is easy to see that $\underset{\thicksim}{\sqsubset}_{RS}$ itself is a stable ready simulation relation and both $\underset{\thicksim}{\sqsubset}_{RS}$ and $\sqsubseteq_{RS}$ are pre-order.
\end{mydefn}

\subsection{The calculus $\text{CLL}_R$ and its operational semantics}

This subsection introduces the LLTS-oriented process calculus $\text{CLL}_R$ presented in \cite{Zhang14}.
Let $V_{AR}$ be an infinite set of variables.
The terms of $\text{CLL}_{R}$ can be given by the following BNF grammar
\[ t::= 0\;|\perp\;|\;(\alpha.t) \;|\; (t\Box t)\;|\;(t\wedge t)\;|\;(t\vee t)\;|\;(t\parallel_A t)\;|\;X\; | \;\langle Z|E \rangle \]
where $X \in V_{AR}$, $\alpha\in Act_\tau$, $A\subseteq Act$ and recursive specification
$E = E(V)$ with $V \subseteq V_{AR}$ is a set of equations $\{X = t| X \in V\}$ and $Z$ is a variable in $V$ that acts as the initial variable.

Most of these operators are from CCS \cite{Milner89} and CSP \cite{Hoare85}:
0 is the process capable of doing no action;
$\alpha.t$ is action prefixing;
$\Box$ is non-deterministic external choice;
$\parallel_A $ is a CSP-style parallel composition.
$\bot$ represents an inconsistent process with empty behavior.
$\vee$ and $\wedge$ are logical operators, which are intended for describing logical combinations of processes.

For any term $\langle Z|E \rangle$ with $E=E(V)$, each variable in  $V$ is bound with scope $E$.
This induces the notion of free occurrence of variable, bound (and free) variables and $\alpha$-equivalence as usual.
A term $t$ is a \emph{process} if it is closed, that is, it contains no free variable.
 %\footnote{It is defined as usual.}.
The set of all processes is denoted by $T(\Sigma_{\text{CLL}_R} )$.
Unless noted otherwise we use $p,q,r$ to represent processes.
Throughout this paper, as usual, we assume that recursive variables are distinct from each other and no recursive variable has free occurrence; moreover we don't distinguish between $\alpha$-equivalent terms and use $\equiv$ for both syntactical identical and $\alpha$-equivalence.
In the sequel, we often denote $\langle X|\{X=t_X\}\rangle$ briefly by $\langle X|X=t_X \rangle$.

 For any recursive specification $E(V)$ and term $t$, the term $\langle t|E \rangle$ is obtained from $t$ by simultaneously replacing all free occurrences of each $X(\in V)$ by $\langle X|E \rangle$, that is,  $\langle t|E \rangle \equiv t\{\langle X|E \rangle/X: X \in V\}$.
For example, consider $t \equiv X \Box a.\langle Y | Y = X \ \Box Y \rangle$ and $E(\{X\})=\{X=t_X\}$ then $\langle t| E\rangle \equiv \langle X|X =t_X\rangle \Box a.\langle Y | Y = \langle X|X=t_X\rangle \Box Y \rangle$.
In particular, for any $E(V)$ and $t \equiv X$, $\langle t|E \rangle \equiv \langle X|E\rangle$ whenever $X \in V$ and $\langle t|E \rangle \equiv X$ if $X \notin V$.

An occurrence of $X$ in $t$ is strongly (or, weakly) guarded if such occurrence is within some subexpression $a.t_1$ with $a \in Act$ ($\tau.t_1$ or $t_1 \vee t_2$ resp.).
A variable $X$ is strongly (or, weakly) guarded in $t$ if each occurrence of $X$ is strongly (weakly resp.) guarded.
A recursive specification $E(V)$ is guarded if for each $X \in V$ and $Z = t_Z \in E(V)$, each occurrence of $X$ in $t_Z$ is (weakly or strongly) guarded.
 As usual, we assume that all recursive specifications considered in the remainder of this paper are guarded.

SOS rules of $\text{CLL}_R$ are listed in Table~\ref{Ta:OPERATIONAL_RULES}, where $a \in Act$, $\alpha \in Act_{\tau}$ and $A \subseteq Act$.
All rules are divided into two parts:

Operational rules specify behaviours of processes.
Negative premises in Rules $Ra_2$, $Ra_3$, $Ra_{13}$ and $Ra_{14}$ give $\tau$-transition precedence over visible transitions, which guarantees that the transition model of $\text{CLL}_{R}$ is $\tau$-pure.
Rules $Ra_9$ and $Ra_{10}$ illustrate that the operational aspect of $t_1\vee t_2$ is same as internal choice in usual process calculus.
Rule $Ra_6$ reflects that conjunction operator is a synchronous product for visible transitions.
The operational rules of the other operators are as usual.

Predicate rules specify  the inconsistency predicate $F$.
Rule $Rp_1$ says that $\bot$ is inconsistent.
Hence $\bot$ cannot be implemented.
While $0$ is consistent and implementable.
Thus $0$ and $\bot$ represent different processes.
Rule $Rp_3$ reflects that if both two disjunctive parts are inconsistent then so is the disjunction.
Rules $Rp_4-Rp_9$ describe the system design strategy that if one part is inconsistent, then so is the whole composition.
Rules $Rp_{10}$ and $Rp_{11}$ reveal that a stable conjunction is inconsistent whenever its conjuncts have distinct ready sets.
Rules $Rp_{13}$ and $Rp_{15}$ %\footnote{Notice that the transition relation $\stackrel{\epsilon}{\Rightarrow}|$ occurring in these two rules does not involve any requirement on consistency, see footnote~1. }
are used to capture (LTS2) in Def.~\ref{D:LLTS}.
Intuitively, these two rules say that if all stable $\tau$-descendants of $z$  are inconsistent, then $z$ itself is inconsistent.

\begin{table}[ht]
%\begin{center}
\rule{\textwidth}{0.5pt}
\noindent \textbf{Operational rules}\\
    $\begin{array}{lll}
    \displaystyle  \quad  Ra_1\frac{-}{\alpha.x_1 \stackrel{\alpha}{\rightarrow} x_1}  &
    \displaystyle  \quad  Ra_2\frac{x_1 \stackrel{a}{\rightarrow} y_1, x_2 \not \stackrel{\tau}{\rightarrow}}{x_1 \Box x_2 \stackrel{a}{\rightarrow} y_1} &
    \displaystyle  \;\;  Ra_3\frac{x_1 \not\stackrel{\tau}{\rightarrow} , x_2 \stackrel{a}{\rightarrow} y_2 }{x_1 \Box x_2 \stackrel{a}{\rightarrow} y_2} \\
    \displaystyle  \quad Ra_4\frac{x_1 \stackrel{\tau}{\rightarrow} y_1}{x_1 \Box x_2 \stackrel{\tau}{\rightarrow} y_1 \Box x_2}&
    \displaystyle   \quad Ra_5\frac{x_2 \stackrel{\tau}{\rightarrow} y_2}{x_1 \Box x_2 \stackrel{\tau}{\rightarrow} x_1 \Box y_2}&
    \displaystyle   \;\; Ra_6\frac{x_1 \stackrel{a}{\rightarrow} y_1, x_2 \stackrel{a}{\rightarrow}y_2}{x_1 \wedge x_2 \stackrel{a}{\rightarrow} y_1 \wedge y_2}\\
    \displaystyle   \quad Ra_7\frac{x_1 \stackrel{\tau}{\rightarrow} y_1}{x_1 \wedge x_2 \stackrel{\tau}{\rightarrow} y_1 \wedge x_2} &
    \displaystyle   \quad Ra_8\frac{x_2 \stackrel{\tau}{\rightarrow} y_2}{x_1 \wedge x_2 \stackrel{\tau}{\rightarrow} x_1 \wedge y_2}&
    \end{array}$

    $\begin{array}{ll}
    \displaystyle  \quad Ra_9\frac{-}{x_1 \vee x_2 \stackrel{\tau}{\rightarrow} x_1}
    &
    \displaystyle  \quad Ra_{10}\frac{-}{x_1 \vee x_2 \stackrel{\tau}{\rightarrow} x_2} \\
    \displaystyle  \quad Ra_{11}\frac{x_1 \stackrel{\tau}{\rightarrow} y_1}{x_1 \parallel_A x_2 \stackrel{\tau}{\rightarrow} y_1\parallel_A x_2} &
    \displaystyle  \quad Ra_{12}\frac{x_2 \stackrel{\tau}{\rightarrow} y_2}{x_1 \parallel_A x_2 \stackrel{\tau}{\rightarrow} x_1 \parallel_A y_2} \\
    \displaystyle  \quad Ra_{13}\frac{x_1 \stackrel{a}{\rightarrow} y_1 , x_2 \not \stackrel{\tau}{\rightarrow} }{x_1 \parallel_A x_2 \stackrel{a}{\rightarrow} y_1 \parallel_A x_2}(a\notin A)&
    \displaystyle  \quad Ra_{14}\frac{x_1 \not\stackrel{\tau}{\rightarrow} , x_2 \stackrel{a}{\rightarrow} y_2 }{x_1 \parallel_A x_2 \stackrel{a}{\rightarrow} x_1 \parallel_A y_2}(a\notin A)\\
    \displaystyle  \quad Ra_{15}\frac{x_1 \stackrel{a}{\rightarrow} y_1, x_2 \stackrel{a}{\rightarrow}y_2}{x_1\parallel_A x_2 \stackrel{a}{\rightarrow} y_1 \parallel_A y_2} (a\in A)&
    \displaystyle  \quad Ra_{16}\frac{\langle t_X| E \rangle  \stackrel{\alpha}{\rightarrow} y}{\langle X|E \rangle \stackrel{\alpha}{\rightarrow} y}(X=t_X \in E)
     \end{array}$

$\;$\\

\noindent \textbf{Predicative rules} \\
$\begin{array}{lll}
       \displaystyle \qquad Rp_1\frac{-}{\bot F}&
       \displaystyle \qquad\qquad Rp_2\frac{x_1 F}{\alpha .x_1 F} &
       \displaystyle \qquad\qquad Rp_3\frac{x_1 F, x_2 F}{x_1\vee x_2 F}\\
       \displaystyle  \qquad Rp_4\frac{x_1 F}{x_1\Box x_2 F}&
       \displaystyle \qquad\qquad Rp_5\frac{x_2 F}{x_1\Box x_2 F}&
       \displaystyle \qquad\qquad Rp_6\frac{x_1 F}{x_1\parallel_A x_2 F}\\
       \displaystyle \qquad Rp_7\frac{x_2 F}{x_1\parallel_A x_2 F}&
       \displaystyle \qquad\qquad Rp_8\frac{x_1 F}{x_1\wedge x_2 F}&
       \displaystyle \qquad\qquad Rp_9\frac{x_2 F}{x_1\wedge x_2 F}
    \end{array}$

    $\begin{array}{ll}
       \displaystyle \quad Rp_{10}\frac{x_1 \stackrel{a}{\rightarrow} y_1, x_2 \not\stackrel{a}{\rightarrow}, x_1 \wedge x_2 \not\stackrel{\tau}{\rightarrow}}{x_1 \wedge x_2 F}&
       \displaystyle \quad  Rp_{11}\frac{x_1 \not\stackrel{a}{\rightarrow} , x_2 \stackrel{a}{\rightarrow} y_2, x_1 \wedge x_2 \not\stackrel{\tau}{\rightarrow}}{x_1 \wedge x_2 F}\\
       \displaystyle \quad Rp_{12}\frac{x_1 \wedge x_2 \stackrel{\alpha}{\rightarrow} z, \{yF:x_1 \wedge x_2 \stackrel{\alpha}{\rightarrow}y\}}{x_1 \wedge x_2 F} &
       \displaystyle \quad Rp_{13}\frac{\{yF:x_1 \wedge x_2 \stackrel{\epsilon}{\Rightarrow}|y\}}{x_1 \wedge x_2 F} \\
       \displaystyle \quad Rp_{14}\frac{\langle t_X|E \rangle F}{\langle X|E \rangle F}(X = t_X \in E) &
       \displaystyle \quad Rp_{15}\frac{\{yF:\langle X|E \rangle \stackrel{\epsilon}{\Rightarrow}|y\}}{\langle X|E \rangle F}
    \end{array}
    $
\rule{\textwidth}{0.5pt}
\caption{SOS rules of $\text{CLL}_R$\label{Ta:OPERATIONAL_RULES}}
%\end{center}
\end{table}

 It has been shown that $\text{CLL}_R$ has the unique stable transition model $M_{\text{CLL}_R}$ \cite{Zhang14},
 which exactly consists of all positive literals of the form $t \stackrel{\alpha}{\rightarrow}t'$ or $tF$ that are provable in $Strip(\text{CLL}_{R},M_{\text{CLL}_{R}})$.
 Here $Strip(\text{CLL}_{R},M_{\text{CLL}_{R}})$ is the stripped version \cite{Bol96} of $\text{CLL}_{R}$ w.r.t $M_{\text{CLL}_{R}}$.
 Each rule in $Strip(\text{CLL}_{R},M_{\text{CLL}_{R}})$ is of the form $\frac{pprem(r)}{conc(r)}$ for some ground instance  $r$ of rules in $\text{CLL}_R$ such that $M_{\text{CLL}_{R}}\models nprem(r)$, where  $nprem(r)$ (or, $pprem(r)$) is the set of negative (positive resp.) premises of $r$, $conc(r)$ is the conclusion of $r$ and $M_{\text{CLL}_{R}}\models nprem(r)$ means that for each $t\not\stackrel{\alpha}{\rightarrow} \in nprem(r)$, $t\stackrel{\alpha}{\rightarrow}s \notin M_{\text{CLL}_{R}}$ for any $s \in T(\Sigma_{\text{CLL}_{R}})$.

   The LTS associated with $\text{CLL}_{R}$, in symbols $LTS(\text{CLL}_{R})$, is the quadruple
    $(T(\Sigma_{\text{CLL}_{R}}),Act_{\tau},\rightarrow_{\text{CLL}_{R}},F_{\text{CLL}_{R}})$, where
    $p \stackrel{\alpha}{\rightarrow}_{\text{CLL}_{R}} p'$ iff $p\stackrel{\alpha}{\rightarrow} p' \in M_{\text{CLL}_{R}}$, and
    $p\in F_{\text{CLL}_{R}}$ iff $pF \in M_{\text{CLL}_{R}}$.
    Therefore $p \stackrel{\alpha}{\rightarrow}_{\text{CLL}_{R}} p'$ (or, $p \in F_{\text{CLL}_R}$) iff  $Strip( \text{CLL}_{R}, M_{\text{CLL}_{R}}) \vdash p\stackrel{\alpha}{\rightarrow}p'$ ($pF$ resp.) for any $p$, $p'$ and $\alpha \in Act_{\tau}$.
    For simplification, in the following we omit the subscripts in $\stackrel{\alpha}{\rightarrow}_{\text{CLL}_{R}}$ and $F_{\text{CLL}_{R}}$.

We end this section by quoting some results from \cite{Zhang14}.

\begin{lemma}\label{L:F_NORMAL}
Let $p$ and $q$ be any two processes. Then

    \noindent (1) $p \vee q \in F $  iff $p,q \in F $;\\
    \noindent (2)  $\alpha.p \in F $ iff $p \in F $ for each $\alpha \in Act_{\tau}$;\\
    \noindent (3)  $p \odot q \in F $  iff either $p \in F $ or $q \in F $  with $\odot \in \{\Box, \parallel_A\}$;\\
    \noindent (4) $p \in F $ or $q \in F $ implies $p \wedge q \in F $;\\
    \noindent (5) $0 \notin F $ and $\bot \in F $.
\end{lemma}

 \begin{theorem}\label{L:LLTS}
    $LTS({\text{CLL}_{R}})$ is a $\tau$-pure LLTS. Moreover if $p\in F$ and $\tau\in \mathcal{I}(p)$ then $\forall q(p\stackrel{\tau}{\rightarrow}q \;\text{implies}\; q \in F)$.
\end{theorem}

\begin{theorem}[precongruence]\label{L:precongruence}
  If $p \sqsubseteq_{RS} q$ then $C_X\{p/X\} \sqsubseteq_{RS} C_X\{q/X\}$, where $C_X$ is any context defined as usual.
\end{theorem}

\section{Axiomatic system $AX_{\text{CLL}}$ and its soundness}

This section is devoted to formulating an axiomatic system for the precongruence $\sqsubseteq_{RS}$ and proving its soundness.
For the moment, we don't know whether a ground-complete proof system exists for the full calculus $\text{CLL}_R$.
This paper will restrict itself to the finite fragment, i.e., leave out recursive operator.

\subsection{$AX_{\text{CLL}}$}

Since inconsistency predicate $F$ (more precisely, $F_{\text{CLL}_R}$) is involved in the definition of $\sqsubseteq_{RS}$, it could be expected that some algebraic laws hold only for processes satisfying certain conditions concerning consistency.
However, since $F$ itself is in semantic category, it is illegal that formulating these conditions in terms of $F$ in axiomatic systems.
Therefore, in order to introduce the axiomatic system $AX_{\text{CLL}}$, a few preliminary definitions are given below, which are needed to express side conditions of some axioms.

\begin{mydefn}[Basic Process Term] \label{D:BPT}
The basic process terms are defined by BNF $t::=0\;|\;(\alpha.t)\;|\;(t\vee t)\;|\;(t \Box t)\;|\;(t\parallel_A t)$, where $\alpha \in Act_{\tau}$ and $A \subseteq Act$. We denote $T(\Sigma_B)$ as the set of all basic process terms.
\end{mydefn}

At a later stage, we will see that the set $T(\Sigma_B)$ is sufficiently expressive to describe all consistent processes with finite behaviours modulo $=_{RS}$.
Moreover, through referring $T(\Sigma_B)$, we can formulate syntactically algebraic laws that hold conditionally, e.g., Axioms $DS4$ and $EXP2$.

\begin{rmk}\label{R:EX_BPT}
  Since all proofs in this section does not depend on the finiteness of processes' behaviour, all results given in this section are still valid if we extend $T(\Sigma_B)$ by adding the item $\langle X|E \rangle$ in BNF above, where $\langle X|E \rangle$ is any strongly guarded processes in $T(\Sigma_{\text{CLL}_R})$ in which neither conjunction operator nor $\bot$ occurs.
  We denote $ET(\Sigma_B)$ as the set of all process terms generating by such extended BNF.
 % for each $Y=t_Y \in E$, $t_Y$ is a term in this expanded language. Such expansion may be useful when considering axiomatic system for full calculus.
  For the purpose of this paper $T(\Sigma_B)$ is sufficient.
%  Here the definition of BPT is enough for our purpose in this paper. If we expand our proof system to the full calculus $\text{CLL}_R$,  the expansion of BPT (addition of recursion) with specific feature (see Lemma~\ref{L:BPT}) is needed.
\end{rmk}

By Lemma~\ref{L:F_NORMAL}, it is easy to see that the operators $\alpha.()$, $\vee$, $\Box$ and $\parallel_A$ preserve consistency. Thus an immediate consequence of Lemma~\ref{L:F_NORMAL} is

\begin{lemma}\label{L:BPT}
  $T(\Sigma_B)\cap F =\emptyset$.
\end{lemma}

Let $<t_0,t_1,\dots,t_{n-1}>$ be a finite sequence of process terms with $n \geq 0$. We define the general external choice $\underset{i<n}\square t_i$ by recursion:
\[ \underset{i<0}\square t_i \triangleq 0,
   \underset{i<1}\square t_i \triangleq t_0,\;\text{and}\;
   \underset{i<k+1}\square t_i \triangleq (\underset{i<k}\square t_i) \Box t_k \;\text{for}\; k \geq 1.\]
Moreover, given a finite sequence $<t_0,\dots,t_{n-1}>$ and $S \subseteq \{t_0,\dots,t_{n-1}\}$, the general external choice $\square S$ is defined as $\square S \triangleq \underset{j<|S|}\square t_j'$, where the sequence $<t_0',\dots,t_{|S|-1}'>$ is the restriction of $<t_0,\dots,t_{n-1}>$ to $S$.
In fact, up to $=_{RS}$ (or, =, see below), the order and grouping of terms in $\underset{i<n}{\square}t_i$ may be ignored by virtue of commutative and associative laws of $\Box$ w.r.t $=_{RS}$ (axioms $EC1$ and $EC2$ below, resp.).

\begin{mydefn}[Injective in Prefixes]
  A process $\underset{i<n}{\square}\alpha_i.t_i$ is injective in prefixes if $\alpha_i \neq \alpha_j$ for any $i \neq j < n$.
\end{mydefn}

The axiomatic system $AX_{\text{CLL}}$ is reported in Table~\ref{Ta:PROOF_SYSTEM_AXIOMS}.
It is an inequational logic where $t=t'$ means $t\leqslant t'$ and $t' \leqslant t$.
Axioms in $AX_{\text{CLL}}$ may be divided into two groups:

\begin{table}
\rule{\textwidth}{0.5pt}
\noindent \textbf{Axioms}
\begin{align*}
       EC1   &  \;\; x \Box y  = y \Box x  &  DI1  &   \;\;  x \vee y  = y \vee x \\
       EC2   &  \;\; (x \Box y)\Box z  = x \Box (y \Box z)       &    DI2   &  \;\; x \vee (y \vee z ) = (x\vee y) \vee z\\
     EC3   &   \;\; x \Box x  = x    &   DI3  &  \;\;   x \vee x  = x        \\
     EC4   &  \;\; x \Box 0  = x    &   DI4   &   \;\;  x \vee \bot  = x    \\
     EC5   &  \;\; x \Box  \bot = \bot    &   DI5  &   \;\;  x  \leqslant x \vee y        \\
     CO1   &  \;\; x\wedge y  = y \wedge x   &  DS1  &  \;\;   x \Box (y\vee z)  \leqslant (x \Box y) \vee (x \Box z)      \\
     CO2   &  \;\; x\wedge x  = x &   DS2  &   \;\;  x \wedge (y\vee z)  \leqslant (x \wedge y) \vee (x \wedge z) \\
     CO3   &  \;\; x\wedge \bot  = \bot    &  DS3  &  \;\;   x \parallel_A (y\vee z)  \leqslant (x \parallel_A y) \vee (x \parallel_A z) \\
     PR1    &  \;\; a.\bot  = \bot    &   DS4  &  \;\;   a.(x \vee y)  \leqslant a.x\Box a.y, \;\text{where}\;  x,y \in T(\Sigma_B)\qquad \\
    PR2  &   \;\;  \tau.x  = x    &    PA1   &   \;\; x \parallel_A y  =y \parallel_A x \\
       &          &    PA2  &   \;\;  x \parallel_A \bot  = \bot
    \end{align*}
\begin{align*}
  ECC1 &\;\;\underset{i< n}{\square}a_i.x_i\wedge \underset{j< m}{\square}b_j.y_j  = \bot \;\text{ if}\; \{a_i|i< n\}\neq \{b_j|j < m\}  \\
   ECC2 &\;\; \underset{i<n}{\square}a_i.(x_i \wedge y_i)  \leqslant \underset{i< n}{\square}a_i.x_i\wedge \underset{i< n}{\square}a_i.y_i  \\
  ECC3& \;\; \underset{i< n}{\square}a_i.x_i\wedge \underset{i< n}{\square}a_i.y_i \leqslant   \underset{i<n}{\square}a_i.(x_i \wedge y_i) \;\text{if}\;\underset{i<n}{\square}a_i.x_i\;\text{is injective in prefixes}  \quad\;\;
\end{align*}
\noindent $EXP1  $
    \begin{multline*}
     \underset{i< n}{\square}a_i.x_i \parallel_A \underset{j< m}{\square}b_j.y_j  \leqslant\\
      \left(\underset {\begin{subarray}
                   \;i< n,\\
                   a_i \notin A
                \end{subarray}}
                \square a_i.(x_i \parallel_A \underset{j< m}{\square}b_j.y_j) \Box
                \underset {\begin{subarray}
                   \;j< m,\\
                   b_j \notin A
                \end{subarray}}
                \square
                b_j.(\underset{i< n}{\square}a_i.x_i \parallel_A y_j)\right) \Box
        \underset {\begin{subarray}
                   \;i< n,j<m\\
                  a_i= b_j\in A
                \end{subarray}}
                \square
                a_i.(x_i \parallel_A y_j)
    \end{multline*}
     $EXP2  \qquad$
      \begin{multline*}
      \left(\underset {\begin{subarray}
                   \;i< n,\\
                   a_i \notin A
                \end{subarray}}
                \square a_i.(x_i \parallel_A \underset{j< m}{\square}b_j.y_j) \Box
                \underset {\begin{subarray}
                   \;j< m,\\
                   b_j \notin A
                \end{subarray}}
                \square
                b_j.(\underset{i< n}{\square}a_i.x_i \parallel_A y_j)\right) \Box
       \underset {\begin{subarray}
                   \;i< n,j<m\\
                  a_i= b_j\in A
                \end{subarray}}
                \square
                a_i.(x_i \parallel_A y_j)\\
     \leqslant \underset{i< n}{\square}a_i.x_i \parallel_A \underset{j< m}{\square}b_j.y_j, \;\text{where}\;  x_i,y_j \in T(\Sigma_B)\; \text{for each}\; i< n\; \text{and}\; j< m
    \end{multline*}
    \noindent \textbf{Inference rules}
    \begin{align*}\hfill
    &\text{REF}  &       &\frac{-}{t\leqslant t}  \\
    &\text{TRANS}  &    &\frac{t\leqslant t',t'\leqslant t''}{t\leqslant t''}  \\
    &\text{CONTEXT}  &     & \text{for each n-ary operator} \;f\\
    & & & \frac{t_1 \leqslant t_1',\dots, t_n \leqslant t_n'}{f(t_1, \dots, t_n) \leqslant f(t_1',\dots,t_n')}
    \end{align*}
    \rule{\textwidth}{0.5pt}
\caption{Axioms and inference rules of $AX_{\text{CLL}}$\label{Ta:PROOF_SYSTEM_AXIOMS}}
\end{table}

%\begin{table}
%\rule{\textwidth}{0.5pt}
%\noindent \textbf{Inference rules}
%    \begin{align*}\hfill
%    &\text{REF}  &       &\frac{-}{t\leqslant t}  \\
%    &\text{TRANS}  &    &\frac{t\leqslant t',t'\leqslant t''}{t\leqslant t''}  \\
%    &\text{CONTEXT}  &     & \text{for each n-ary operator} \;f\\
%    & & & \frac{t_1 \leqslant t_1',\dots, t_n \leqslant t_n'}{f(t_1, \dots, t_n) \leqslant f(t_1',\dots,t_n')}
%    \end{align*}
%\rule{\textwidth}{0.5pt}
%\caption{Inference rules of $AX_{\text{CLL}}$\label{Ta:PROOF_SYSTEM_RULES}}
%\end{table}

 First the ones that involve only a single operator, which capture fundamental properties of operators, e.g., commutativity, associativity, idempotent, etc.
      These axioms are standard.

 Second the ones that characterize the interaction between operators.
        Among them, the axioms $DS1$, $DS3$, $DS4$ and $ECCi(1 \leq i \leq 3)$ describe the interaction between logical and operational operators.
        As mentioned early, it is one distinguishing feature of LLTS that it involves consideration of inconsistencies.
A number of axioms in this group embody such feature.
In particular, as a consequence of considering inconsistency, side conditions are associated with $DS4$, $ECC3$ and $EXP2$.
In the next subsection, we will show that these side conditions are necessary by giving counterexamples.

It should be pointed out that some axioms have been considered by L\"{u}ttgen and Vogler semantically in \cite{Luttgen10}, including $DS2$, $CO2$, $CO3$ and $DIi(3\leq i \leq 5)$.

Given the axioms and rules of inference, we assume that the resulting notions of proof, length of proof and theorem are already familiar to the reader. Following standard usage, $\vdash t \leqslant t'$ means that $t \leqslant t'$ is a theorem of $AX_{\text{CLL}}$.

%
%We end this subsection by discussing axioms of $AX_{\text{CLL}}$.
%Many axioms look familiar and appear in the axiomatizations of semantics in other process calculi.
%However, since $\text{CLL}_R$ considers inconsistencies, it is not trivial to find many axioms such as Axioms~$DS4$, $ECC3$ and $EXP2$, where the additional conditions help us to treat inconsistencies smoothly.

\subsection{Soundness}
This subsection will establish the soundness of $AX_{\text{CLL}}$ w.r.t $\sqsubseteq_{RS}$.
Although $AX_{\text{CLL}}$ is a proof system for $\text{CLL}_R$-processes with finite behaviours, it is sound for the full calculus.
Therefore this subsection doesn't restrict itself to finite terms.

%First we focus on laws concerning one operator.
% These laws capture inherent properties of operators, e.g., commutativity, associativity and zero element.
% Some laws obtained by L\"{u}ttgen and Vogler in \cite{Luttgen10} will also be rephrased in process-algebraic style.
% Before giving these laws, we show next two lemmas.

As usual, in order to get soundness, we need to check that all ground instances of axioms are sound w.r.t $\sqsubseteq_{RS}$ and all inference are sound.
The latter immediately follows from reflexivity and transitivity of $\sqsubseteq_{RS}$ and Theorem~\ref{L:precongruence}.
Therefore the remainder of this subsection will devote itself to verifying the soundness of axioms.

 We begin by giving a simple but useful property about combined processes $p \odot q$ with $\odot \in \{\Box,\parallel_A,\wedge\}$.
 Roughly speaking, it says that consistent and stable $\epsilon$-derivatives of $p \odot q$ must be compositions of consistent and stable $\epsilon$-derivatives of $p$ and $q$, and the converse also (almost) holds.

 \begin{lemma}\label{L:TAU_I}
  \noindent  (1) For any $\odot \in \{\Box,\parallel_A,\wedge\}$, if $p_1 \odot p_2 \stackrel{\epsilon}{\Rightarrow}_F| p_3$ then $p_1 \stackrel{\epsilon}{\Rightarrow}_F| p_1'$, $p_2 \stackrel{\epsilon}{\Rightarrow}_F|p_2'$ and $p_3 \equiv p_1' \odot p_2'$ for some $p_1', p_2'$.

  \noindent (2)     If $p_1 \stackrel{\epsilon}{\Rightarrow}_F| p_1'$ and $p_2 \stackrel{\epsilon}{\Rightarrow}_F|p_2'$ then $p_1 \odot p_2 \stackrel{\epsilon}{\Rightarrow}_F| p_1' \odot p_2'$  for $\odot \in \{\Box,\parallel_A\}$,
            and $p_1\wedge p_2 \stackrel{\epsilon}{\Rightarrow}_F| p_1' \wedge p_2'$ if $p_1' \wedge p_2' \notin F$.
\end{lemma}
\begin{proof}
  Straightforward by applying Theorem~\ref{L:LLTS} and Lemma~\ref{L:F_NORMAL}.
\end{proof}

The next observation, which is due to L\"{u}ttgen and Vogler, reveals that the relation $\sqsubseteq_{RS}$ interacts well with logic operators conjunction and disjunction.

\begin{lemma}\label{L:CON_ID_I}

\noindent (1) $p_i \sqsubseteq_{RS} p_1 \vee p_2$ for $i=1,2$.

\noindent (2) If $p_1 \sqsubseteq_{RS} p_3 $ and $p_2 \sqsubseteq_{RS} p_3$ then $p_1 \vee p_2\sqsubseteq_{RS} p_3 $.

\noindent (3) $p_1 \wedge p_2 \sqsubseteq_{RS} p_i$ for $i=1,2$.

\noindent (4) If $p_1 \sqsubseteq_{RS} p_2 $ and $p_1 \sqsubseteq_{RS} p_3$, then $p_1  \sqsubseteq_{RS} p_2 \wedge p_3$.

\end{lemma}
\begin{proof}

\noindent \textbf{(1,2)} Straightforward.

\noindent \textbf{(3)}   Assume $p_1 \wedge p_2 \stackrel{\epsilon}{\Rightarrow}_F| p_{12}$.
    By Lemma~\ref{L:TAU_I}, $p_1 \stackrel{\epsilon}{\Rightarrow}_F| p_1'$ and $p_{12} \equiv p_1' \wedge p_2'$ for some $p_1',p_2'$.
    Then it suffices to show $ p_1'  \wedge p_2' \underset{\thicksim}{\sqsubset}_{RS} p_1'$.
    To this end, put ${\mathcal R}\triangleq\{(s \wedge t , s)| \;s\;\text{and}\;t \;\text{are stable}\}$.
    It is routine to verify that $\mathcal R$ is a stable ready simulation relation, as desired.

\noindent \textbf{(4)} It immediately follows from Lemma~\ref{L:TAU_I} and the fact that $p  \underset{\thicksim}{\sqsubset}_{RS} q$ and $p \underset{\thicksim}{\sqsubset}_{RS}r$ implies $p  \underset{\thicksim}{\sqsubset}_{RS} q \wedge r$ (see \cite[Lemma~4.5]{Zhang14}).
%, we have proved that if .
%Then it is not difficult to prove item (2) by using this result and Lemma~\ref{L:TAU_I}.
%
%Next we handle $p_1 \wedge p_2 \sqsubseteq_{RS} p_1$ in item (1).
%First we prove $p_1  \wedge p_2 \underset{\thicksim}{\sqsubset}_{RS} p_1$ where $p_1,p_2$ are stable.
%    Put ${\mathcal R}\triangleq\{(s \wedge t , s)| \;s\;\text{and}\;t \;\text{are stable}\}$.
%   It suffices to prove that $\mathcal R$ is a stable ready simulation relation.
%    Let $(p_1 \wedge p_2 , p_1) \in {\mathcal R}$.
%    Clearly $(p_1 \wedge p_2 , p_1)$ satisfies (RS1) and (RS2).
%
%    \textbf{(RS3)} Suppose $p_1 \wedge p_2 \stackrel{a}{\Rightarrow}_F| p_{12}$.
%    Then $p_1 \wedge p_2 \stackrel{a}{\rightarrow}_F p_{12}' \stackrel{\epsilon}{\Rightarrow}_F|p_{12}$ for some $p_{12}'$.
%    So, by Lemma~\ref{L:TAU_I}, we get $p_i  \stackrel{a}{\rightarrow}_F p_i' \stackrel{\epsilon}{\Rightarrow}_F|p_i''$ for some $p_i',p_i''$ with $i=1,2$, $p_{12}\equiv p_1'' \wedge p_2''$ and $(p_{12},p_1'') \in {\mathcal R}$.
%
%    \textbf{(RS4)} Suppose $p_1 \wedge p_2 \notin F $. By Rules~$Rp_{10}$ and $Rp_{11}$,  we have ${\mathcal I}(p_1 ) = {\mathcal I}(p_2)$. Clearly ${\mathcal I}(p_1 \wedge p_2) = {\mathcal I}(p_1)$.
\end{proof}

As an immediate consequence of items (3) and (4) in previous lemma, the property below is given, which is obtained in \cite{Luttgen10}.
\[p_1 \sqsubseteq_{RS} p_2 \wedge p_3 \;\text{iff}\;p_1 \sqsubseteq_{RS} p_2\;\text{and}\; p_1 \sqsubseteq_{RS} p_3. \tag{FP}\]

As pointed out by L\"{u}ttgen and Vogler \cite{Luttgen07,Luttgen10}, this is a fundamental property of ready simulation in the presence of logic operators. Intuitively, it says that $p_1$ is an implementation  of the specification $p_2\wedge p_3$ if and only if $p_1$ implements both $p_2$ and $p_3$.
Moreover, by Lemma~\ref{L:CON_ID_I}, it is easy to see that the following equation holds.
\[p \wedge (p \vee q)=_{RS}p=_{RS}p \vee (p \wedge q)\tag{Absorption}\]

More fundamental algebraic laws are collected in the next proposition.

\begin{proposition}\label{L:ONE_OPERATOR}\hfill

\noindent (1) Commutativity: $p_1\odot p_2 =_{RS} p_2\odot p_1$ for each $\odot \in \{\Box,\parallel_A,\wedge,\vee\}$;

\noindent (2) Associativity: $(p_1 \odot p_2) \odot p_3 =_{RS} p_1 \odot (p_2 \odot p_3)$ for each $\odot \in \{\Box,\vee,\wedge\}$;

\noindent (3) Idempotency: $p\odot p =_{RS} p$ for each $\odot \in \{\Box, \wedge,\vee\}$;

\noindent (4) Unit element: $p\Box 0 =_{RS} p$, $p \vee \bot =_{RS}p$;

\noindent (5) Zero element: $p\odot \bot =_{RS} \bot$ for each $\odot \in \{\Box,\parallel_A,\wedge\}$;

\noindent (6) Identity property: $\tau.p=_{RS}p$, $\alpha.\bot=_{RS}\bot$.
\end{proposition}
\begin{proof}
%  We handle $p\wedge q \sqsubseteq_{RS}q \wedge p$ and $p_1 \Box p_2 \sqsubseteq_{RS} p_2 \Box p_1$ as example, others are not difficult to prove and omitted.
%  For the former, by Lemma~\ref{L:CON_ID_I}(1)(2), $p\wedge q \sqsubseteq_{RS}q $ and $p\wedge q \sqsubseteq_{RS} p$, and then $p\wedge q \sqsubseteq_{RS}q \wedge p$.
%
%  For the latter, first we prove $p_1 \Box p_2 \underset{\thicksim}{\sqsubset}_{RS} p_2 \Box p_1$, where $p_1,p_2$ are stable. Let $p_1,p_2$ be two stable process and set ${\mathcal R}\triangleq \{(p_1 \Box p_2, p_2 \Box p_1)\} \cup Id$  \footnote{$Id$ is an identity relation over $T(\Sigma_{\text{CLL}_R})$}. It suffices to prove that $\mathcal R$ is a stable ready simulation relation. It is not difficult to check that each pair in $\mathcal R$ satisfies (RS1)-(RS4).
%  Then we prove $p_1 \Box p_2 \sqsubseteq_{RS} p_2 \Box p_1$. Assume $p_1 \Box p_2 \stackrel{\epsilon}{\Rightarrow}_F|r$. By Lemma~\ref{L:TAU_I}, we get $p_1  \stackrel{\epsilon}{\Rightarrow}_F|p_1'$, $p_2  \stackrel{\epsilon}{\Rightarrow}_F|p_2'$ and $r \equiv p_1' \Box p_2'$ for some $p_1',p_2'$, then $p_2 \Box p_1  \stackrel{\epsilon}{\Rightarrow}_F|p_2' \Box p_1'$ and $p_1' \Box p_2' \underset{\thicksim}{\sqsubset}_{RS} p_2' \Box p_1'$.
We give the proof only for Commutativity laws, the other laws are left to the reader.
Clearly Commutativity laws for $\wedge$ and $\vee$ are implied by Lemma~\ref{L:CON_ID_I}.
For $\odot \in \{\Box,\parallel_A\}$, the argument is similar to that in the proof of Lemma~\ref{L:CON_ID_I}(3), that is, by Lemma~\ref{L:TAU_I}, it is enough to check that the relation ${\mathcal R}_{\odot}$ below is a stable ready simulation relation.
\[{\mathcal R}_{\odot}\triangleq \{(p \odot q, q \odot p):p,q \;\text{are stable}\} \cup Id\]
where $Id$ is the identity relation over $T(\Sigma_{\text{CLL}_R})$.
\end{proof}

\begin{rmk}
  Due to Commutativity, Associativity, Idempotency and Absorption laws of $\wedge$ and $\vee$, modulo $=_{RS}$, the structure $< T(\Sigma_{\text{CLL}_R}),\wedge,\vee>$ is a lattice.
  In fact, such lattice is distributive by Prop.~\ref{S:DISTRIBUTIVE} given later.
  Moreover, by Lemma~\ref{L:CON_ID_I}(3) and (FP), the partial order corresponding to the lattice $< T(\Sigma_{\text{CLL}_R}),\wedge,\vee >$ indeed is $\sqsubseteq_{RS}$, that is, $p\sqsubseteq_{RS} q$ iff  $p\wedge q =_{RS}p$ for any $p,q \in T(\Sigma_{\text{CLL}_R})$.
\end{rmk}

%Hitherto we have only considered (in)equational laws involving one operator alone.
In the following, we shall deal with a few of laws  referring to different operators in one (in)equation.
In order to show so-called distributive law, the next lemma is needed which reveals that there exist  ``canonical'' evolving paths from $p_1 \odot (p_2 \vee p_2)$ to its stable $\epsilon$-derivatives (if exist).

\begin{lemma}\label{L:DIS}
    Let $\odot\in\{\Box,\wedge,\parallel_A\}$. If $p_1 \odot( p_2 \vee p_3) \stackrel{\epsilon}{\Rightarrow}_F|p_4$  then there are $p_1'$ and $r_i(i \leq n\;\text{and}\;n>0)$ such that (1) $p_1 \odot( p_2 \vee p_3) \equiv r_0 \stackrel{\tau}{\rightarrow}_F,\dots,\stackrel{\tau}{\rightarrow}_F r_n \equiv p_4$, (2) $p_1\stackrel{\epsilon}{\Rightarrow}_F p_1'$, (3) $r_j\equiv p_1' \odot( p_2 \vee p_3)$ and  $r_{j+1}\equiv p_1' \odot p_k $ for some $j<n$ and $k\in \{2,3\}$.
\end{lemma}
\begin{proof}
Since $p_1 \odot( p_2 \vee p_3) \stackrel{\epsilon}{\Rightarrow}_F|p_4$ and $p_2 \vee p_3 \stackrel{\tau}{\rightarrow}$, $p_1 \odot( p_2 \vee p_3) (\stackrel{\tau}{\rightarrow}_F)^m|p_4$ for some $m > 0$.
The rest of the proof is routine by induction on $m$.
%Since , we have $n>0$.
%By using Lemma~\ref{L:F_NORMAL}, the proof is straightforward by induction on $n$.
\end{proof}

The following Distributive law with $\odot = \wedge$ was first proved in \cite{Luttgen10}.

\begin{proposition}[Distributive]\label{S:DISTRIBUTIVE}
       $p_1 \odot( p_2 \vee p_3) =_{RS} (p_1 \odot p_2) \vee (p_1 \odot p_3)$ for each  $\odot \in \{ \Box, \parallel_A,\wedge \}$.
\end{proposition}
\begin{proof}
  %  We consider the case $\odot = \Box$, the others are similar.
    The inequation $(p_1 \odot p_2 )\vee (p_1 \odot p_3) \sqsubseteq_{RS} p_1 \odot (p_2 \vee p_3)$ immediately follows from Theorem~\ref{L:precongruence} and Lemma~\ref{L:CON_ID_I}(1)(2).
    For the converse inequation, suppose $p_1 \odot( p_2 \vee p_3) \stackrel{\epsilon}{\Rightarrow}_F| p_4$.
    Then by Theorem~\ref{L:LLTS} and Lemma~\ref{L:DIS}, it is easy to get   $(p_1 \odot p_2) \vee (p_1 \odot p_3) \stackrel{\epsilon}{\Rightarrow}_F| p_4$.
    Hence  $p_1 \odot( p_2 \vee p_3) \sqsubseteq_{RS} (p_1 \odot p_2) \vee (p_1 \odot p_3)$.
\end{proof}

Since $< T(\Sigma_{\text{CLL}_R}),\wedge,\vee >$ is a lattice, it immediately follows from Prop.~\ref{S:DISTRIBUTIVE} with $\odot = \wedge$ that $p_1 \vee (p_2 \wedge p_3)=_{RS}(p_1 \vee p_2)\wedge (p_1 \vee p_3)$.

\begin{proposition}\label{S:SPECIAL_I}
$\alpha.p_1 \Box \alpha.p_2 \sqsubseteq_{RS} \alpha.(p_1 \vee p_2) $ for each $\alpha \in Act_{\tau}$.
\end{proposition}
\begin{proof}
   $ p_1 \sqsubseteq_{RS} p_1 \vee p_2\;\text{and}\;p_2 \sqsubseteq_{RS} p_1 \vee p_2$  \qquad\qquad  (by Lemma~\ref{L:CON_ID_I}(1))\\
 $\Rightarrow \alpha.p_1 \sqsubseteq_{RS} \alpha.(p_1 \vee p_2)\;\text{and}\;\alpha.p_2 \sqsubseteq_{RS} \alpha.(p_1 \vee p_2)$ \quad (by Theorem~\ref{L:precongruence})\\
 $\Rightarrow \alpha.p_1 \Box \alpha.p_2 \sqsubseteq_{RS} \alpha.(p_1 \vee p_2)$ \qquad \qquad  (by Theorem~\ref{L:precongruence} and Prop.~\ref{L:ONE_OPERATOR}).
\end{proof}

A natural problem arises at this point, that is, whether the inequation below holds
\[\alpha.(p_1 \vee p_2) \sqsubseteq_{RS} \alpha.p_1 \Box \alpha.p_2. \tag{DS}\]
The answer is negative by considering $p_1 \equiv \bot$ and $p_2 \equiv 0$.
By Lemma~\ref{L:F_NORMAL}, $a.(\bot \vee 0)\notin F $ and $a.\bot \Box a.0 \in F $.
Hence $a.(\bot \vee 0) \not\sqsubseteq_{RS} a.\bot \Box a.0$.
However we can give a necessary and sufficient condition for the inequation (DS) with $\alpha \in  Act$ to be true.
To this end, we introduce the notion

\begin{mydefn}[Uniform w.r.t $F $]
  Two processes $p$ and $q$ are uniform w.r.t $F $ if $p\in F $ iff $q \in F $.
\end{mydefn}

\begin{proposition}\label{S:SPECIAL}
For each $a \in Act$,
 $a.(p_1 \vee p_2) \sqsubseteq_{RS} a.p_1 \Box a.p_2 $ iff
 $p_1$ and $p_2$ are uniform w.r.t $F $.
\end{proposition}
\begin{proof}
\noindent \textbf{(Left implies Right)} Suppose $p_1$ and $p_2$ are not uniform w.r.t $F $. W.l.o.g, assume that $p_1 \in F $ and $p_2\notin F $. By Lemma~\ref{L:F_NORMAL}, we get $a.(p_1 \vee p_2) \notin F $ and $a.p_1 \Box a.p_2 \in F $. Hence $a.(p_1 \vee p_2) \not\sqsubseteq_{RS} a.p_1 \Box a.p_2 $.

\noindent \textbf{(Right implies Left)}
Since  $a \in Act$,  it suffices to prove  $a.(p_1 \vee p_2) \underset{\thicksim}{\sqsubset}_{RS}  a.p_1 \Box a.p_2 $.
Put \[{\mathcal R}\triangleq\{(a.(p_1 \vee p_2), a.p_1 \Box a.p_2)\} \cup Id.\]
We will show that $\mathcal R$ is a stable ready simulation relation.
It is obvious that (RS1-4) hold for each pair in $Id$.
In the following, we deal with the pair $(a.(p_1 \vee p_2),a.p_1 \Box a.p_2)$.
Clearly, such pair satisfies (RS1) and (RS4) .

\textbf{(RS2)} Suppose $a.p_1 \Box a.p_2 \in F $. By Lemma~\ref{L:F_NORMAL}, $p_i \in F $ for some $i\in \{1,2\}$. Then, since $p_1$ and $p_2$ are uniform w.r.t $F $, we get $p_1,p_2 \in F $. So  $a.(p_1 \vee p_2) \in F $.

\textbf{(RS3)} Suppose $a.(p_1 \vee p_2) \stackrel{a}{\Rightarrow}_F|r$.
It is easy to see that $a.p_1 \Box a.p_2 \stackrel{a}{\rightarrow} \stackrel{\epsilon}{\Rightarrow}_F| r$.
Moreover $a.p_1 \Box a.p_2 \notin F $ by $a.(p_1 \vee p_2)\notin F $ and (RS2).
So $a.p_1 \Box a.p_2 \stackrel{a}{\Rightarrow}_F| r$.
\end{proof}

Notice that the situation is different if $\alpha = \tau$.
In such case, the inequation (DS) does not always hold even if $p_1$ and $p_2$ are  uniform w.r.t $F $.
As a simple example, consider $p_1\equiv a.0$ and $p_2\equiv b.0$ with $a \neq b$.
Clearly, they are uniform w.r.t $F$ because of $p_1,p_2 \notin F$.
Moreover, $\tau.(a.0\vee b.0)\stackrel{\epsilon}{\Rightarrow}_F|a.0$,  and $a.0 \Box b.0$ is the unique process such that $\tau.a.0\Box \tau.b.0 \stackrel{\epsilon}{\Rightarrow}_F|a.0\Box b.0$.
But $a.0 \not \underset{\thicksim}{\sqsubset}_{RS} a.0\Box b.0$ due to $a.0 \notin F $ and ${\mathcal I}(a.0)\neq {\mathcal I}(a.0\Box b.0)$.
Thus $\tau.(a.0\vee b.0) \not\sqsubseteq_{RS} \tau.a.0 \Box \tau.b.0$.

%In the following, we show  a number of properties of interaction of general external choice $\underset{i<n}\square t_i$ and conjunction.  First, a simpler result is given:
Given the key role that general external choice $\underset{i<n}\square p_i$ plays in the axiomatic system $AX_{\text{CLL}}$, we need to discuss this operator in some detail.
We begin with giving the following simple result, of which we omit the straightforward proof.

\begin{lemma}\label{L:BIG_SQUARE}
Let $n\geq 0$ and $\{a_i | i<n \}\subseteq Act$.

\noindent  (1) $ \underset{i<n}\square p_i \in F $ iff $p_k\in F $ for some $k<n$.

\noindent  (2)  $\underset{i<n}\square a_i.p_i \stackrel{a_i}{\rightarrow}  p_i$ for each $i<n$.

\noindent  (3) If $\underset{i<n}\square a_i.p_i \stackrel{\alpha}{\rightarrow}  s$ then $\alpha = a_k$ and $s \equiv p_k$ for some $k<n$.
\end{lemma}

\begin{proposition}\label{L:MULTIPLE_VI}
Let $a_i, b_j \in Act$ for each $i<n$ and $j<m$.

\noindent (1) If $\{a_i|i< n\}\neq \{b_j|j< m\}$ then $\underset{i< n}{\square}a_i.p_i \wedge \underset{j< m }{\square} b_j.q_j =_{RS} \bot$.

\noindent (2) $\underset{i< n} {\square}a_i.(p_i \wedge q_i) \sqsubseteq_{RS} \underset{i< n}{\square}a_i.p_i \wedge \underset{i< n}{\square}a_i.q_i$.
\end{proposition}
\begin{proof}
\textbf{(1)} By Rules $Rp_{10}$ and $Rp_{11}$, it holds trivially.

\noindent \textbf{(2)} If $n=0$, it is trivial because of the definition of general external choice.
Next we treat the case $n>0$.
By Lemma~\ref{L:CON_ID_I} and Theorem~\ref{L:precongruence},
$a_i.(p_i \wedge q_i) \sqsubseteq_{RS} a_i.p_i$ for each $i<n$.
Then $\underset{i< n} {\square}a_i.(p_i \wedge q_i) \sqsubseteq_{RS} \underset{i< n}{\square}a_i.p_i$ by Theorem~\ref{L:precongruence} and Prop.~\ref{L:ONE_OPERATOR}.
Similarly, we also have $\underset{i< n} {\square}a_i.(p_i \wedge q_i) \sqsubseteq_{RS} \underset{i< n}{\square}a_i.q_i$.
Hence $\underset{i< n} {\square}a_i.(p_i \wedge q_i) \sqsubseteq_{RS} \underset{i< n}{\square}a_i.p_i \wedge \underset{i< n}{\square}a_i.q_i$ by Lemma~\ref{L:CON_ID_I}.
\end{proof}

In the following, we provide an example to illustrate that it does not always hold that $\underset{i< n}{\square}a_i.p_i\wedge \underset{i< n}{\square}a_i.q_i \sqsubseteq_{RS} \underset{i< n} {\square}a_i.(p_i \wedge q_i)$.

\begin{example}
Consider process $a_0.p_0 \triangleq a.b.0$, $a_1.p_1 \triangleq a.c.0$, $a_0.q_0 \triangleq a.b.0$ and $a_1.q_1 \triangleq a.b.0$ where $c \neq b$.
Then, $\underset{i<2}\square a_i.p_i \equiv a.b.0 \Box a.c.0$,
$\underset{i<2}\square a_i.q_i \equiv a.b.0 \Box a.b.0$
and $\underset{i<2}\square a_i.(p_i \wedge q_i) \equiv a.(b.0 \wedge b.0) \Box a.(c.0 \wedge b.0)$.
Assume for contradiction that $\underset{i<2}\square a_i.p_i \wedge \underset{i<2}\square a_i.q_i \sqsubseteq_{RS} \underset{i<2}\square a_i.(p_i \wedge q_i) $.
Thus $\underset{i<2}\square a_i.p_i \wedge \underset{i<2}\square a_i.q_i \underset{\thicksim}{\sqsubset}_{RS} \underset{i<2}\square a_i.(p_i \wedge q_i) $ due to $a \in Act$.
It follows from $c.0 \wedge b.0 \not\stackrel{\tau}{\rightarrow}$ and $b.0 \not\stackrel{c}{\rightarrow}$ that $\frac{c.0 \stackrel{c}{\rightarrow}0}{c.0 \wedge b.0F} \in Strip(\text{CLL}_R,M_{\text{CLL}_R})$.
So $c.0 \wedge b.0 \in F$ because of $c.0 \stackrel{c}{\rightarrow}0$.
Further $\underset{i<2}\square a_i.(p_i \wedge q_i)  \in F$ by Lemma~\ref{L:F_NORMAL}.
Thus, it follows from $\underset{i<2}\square a_i.p_i \wedge \underset{i<2}\square a_i.q_i \underset{\thicksim}{\sqsubset}_{RS} \underset{i<2}\square a_i.(p_i \wedge q_i) $ that $\underset{i<2}\square a_i.p_i \wedge \underset{i<2}\square a_i.q_i \in F$.
Since $\underset{i<2}\square a_i.p_i \notin F$, $ \underset{i<2}\square a_i.q_i \notin F$ and ${\mathcal I}(\underset{i<2}\square a_i.p_i) = {\mathcal I}(\underset{i<2}\square a_i.q_i)$, the last rule applied in the proof tree of $Strip(\text{CLL}_R,M_{\text{CLL}_R})\vdash \underset{i<2}\square a_i.p_i \wedge \underset{i<2}\square a_i.q_iF$ is of the form
\[\frac{\{sF:\underset{i<2}\square a_i.p_i \wedge \underset{i<2}\square a_i.q_i \stackrel{a}{\rightarrow}s\}}{\underset{i<2}\square a_i.p_i \wedge \underset{i<2}\square a_i.q_iF}\;\text{or}\; \frac{\{sF:\underset{i<2}\square a_i.p_i \wedge \underset{i<2}\square a_i.q_i \stackrel{\epsilon}{\Rightarrow}|s\}}{\underset{i<2}\square a_i.p_i \wedge \underset{i<2}\square a_i.q_iF}.\]
However, since $b.0 \wedge b.0$ is an $a$-derivative of $\underset{i<2}\square a_i.p_i \wedge \underset{i<2}\square a_i.q_i$ and $b.0 \wedge b.0 \notin F$, the former is impossible.
Moreover, since $\underset{i<2}\square a_i.p_i \wedge \underset{i<2}\square a_i.q_i$ is the unique stable $\epsilon$-derivative of itself, the latter is also impossible due to the well-foundedness of proof tree.
Thus a contradiction arises, as desired.
%Consider process $(a.b.0 \Box a.c.0) \wedge  (a.b.0 \Box a.b.0)$ and $a.(b.0 \wedge b.0) \Box a.(c.0 \wedge b.0)$ with $b \neq c$.
%By Rules~$Rp_{10}$ and $Rp_{11}$, we get $c.0 \wedge b.0 \in F$.
%Then $a.(c.0 \wedge b.0) \in F$ and $a.(b.0 \wedge b.0) \Box a.(c.0 \wedge b.0) \in F$ by Lemma~\ref{L:F_NORMAL}.
%Next we intend to prove that $(a.b.0 \Box a.c.0) \wedge  (a.b.0 \Box a.b.0) \notin F$.
%Contrarily, assume that $(a.b.0 \Box a.c.0) \wedge  (a.b.0 \Box a.b.0) \in F$.
%Consider the last rule applied in the proof tree of
%\[Strip(\text{CLL}_R,M_{\text{CLL}_R})\vdash (a.b.0 \Box a.c.0) \wedge  (a.b.0 \Box a.b.0) F.\tag{EPT}\]
%By Lemma~\ref{L:F_NORMAL}, it is not difficult to see that $a.b.0 \Box a.c.0\notin F$ and $a.b.0 \Box a.b.0\notin F$.
%Moreover ${\mathcal I}(a.b.0 \Box a.c.0)={\mathcal I}(a.b.0 \Box a.b.0)=\{a\}$.
%Hence the last rule applied in (EPT) is $Rp_{12}$ or $Rp_{13}$.
%Since the proof tree is well-founded and $(a.b.0 \Box a.c.0) \wedge  (a.b.0 \Box a.b.0)$ is stable, Rule~$Rp_{13}$ is impossible.
%Then we check Rule~$Rp_{12}$.
%It is easy to see that $b.0 \wedge b.0$ is an $a$-derivative of $(a.b.0 \Box a.c.0) \wedge  (a.b.0 \Box a.b.0)$ and $b.0 \wedge b.0 \notin F$.
%Hence Rule~$Rp_{12}$ is also impossible and a contradiction arises.
%Therefore $(a.b.0 \Box a.c.0) \wedge  (a.b.0 \Box a.b.0) \notin F$.
%Thus $(a.b.0 \Box a.c.0) \wedge  (a.b.0 \Box a.b.0)\not \sqsubseteq_{RS} a.(b.0 \wedge b.0) \Box a.(c.0 \wedge b.0)$.
\end{example}

However, for any $\underset{i<n}\square a_i.p_i$ with distinct prefixes, we have

\begin{proposition}\label{L:MULTIPLE_II}
Let $a_i \in Act$ for each $i<n$.
  If $\underset{i< n}{\square}a_i.p_i$ is injective in prefixes then $\underset{i< n}{\square}a_i.p_i\wedge \underset{i< n}{\square}a_i.q_i \sqsubseteq_{RS} \underset{i< n} {\square}a_i.(p_i \wedge q_i)$.
\end{proposition}
\begin{proof}
We examine the case $n>0$.
Since $\{a_i:i<n\} \subseteq Act$, it suffices to prove $ \underset{i< n}{\square}a_i.p_i \wedge \underset{i< n}{\square}a_i.q_i \underset{\thicksim}{\sqsubset}_{RS} \underset{i< n} {\square}a_i.(p_i \wedge q_i)$.
Put
 \[{\mathcal R}\triangleq \{ (\underset{i< n}{\square}a_i.p_i \wedge \underset{i< n}{\square}a_i.q_i , \underset{i< n} {\square}a_i.(p_i \wedge q_i)) \} \cup Id.\]
We need to check that $(\underset{i< n}{\square}a_i.p_i \wedge \underset{i< n}{\square}a_i.q_i , \underset{i< n} {\square}a_i.(p_i \wedge q_i))$ satisfies (RS1-4).
 For the conditions (RS1,4), it is trivial and omitted.

\textbf{(RS2)} Suppose $\underset{i< n} {\square}a_i.(p_i \wedge q_i) \in F $.
Then, by Lemma~\ref{L:BIG_SQUARE}, $p_k \wedge q_k \in F $ for some $k$.
Since both $\underset{i<n}{\square}a_i.p_i$ and $\underset{i<n}{\square}a_i.q_i$ are injective in prefixes, $p_k \wedge q_k$ is the unique $a_k$-derivative of $ \underset{i< n}{\square}a_i.p_i\wedge \underset{i< n}{\square}a_i.q_i$.
Therefore $ \underset{i< n}{\square}a_i.p_i\wedge \underset{i< n}{\square}a_i.q_i \in F $ comes from $p_k \wedge q_k \in F $ by Theorem~\ref{L:LLTS} and (LTS1) in Def.~\ref{D:LLTS}, as desired.

\textbf{(RS3)} Suppose $\underset{i< n}{\square}a_i.p_i \wedge \underset{i< n}{\square}a_i.q_i \stackrel{a}{\Rightarrow}_F |p'$.
 Then $\underset{i< n}{\square}a_i.p_i \wedge \underset{i< n}{\square}a_i.q_i \stackrel{a}{\rightarrow}_F p''  \stackrel{\epsilon}{\Rightarrow}_F |p'$ for some $p''$.
Since $\underset{i<n}\square a_i.p_i$ and $\underset{i<n} \square a_i.q_i$ are injective in prefixes, there exists $k<n$ such that $\underset{i< n}{\square}a_i.p_i \stackrel{a_k}{\rightarrow}  p_k$, $\underset{i< n}{\square}a_i.q_i \stackrel{a_k}{\rightarrow}  q_k$, $a = a_k$ and $p'' \equiv p_k \wedge q_k$.
Clearly $\underset{i< n} {\square}a_i.(p_i \wedge q_i) \stackrel{a_k}{\rightarrow}  p_k \wedge q_k$.
Moreover $\underset{i< n} {\square}a_i.(p_i \wedge q_i) \notin F $ by $\underset{i< n}{\square}a_i.p_i \wedge \underset{i< n}{\square}a_i.q_i \notin F $ and (RS2).
Hence $\underset{i< n} {\square}a_i.(p_i \wedge q_i) \stackrel{a}{\rightarrow}_F p_k\wedge q_k \equiv p''  \stackrel{\epsilon}{\Rightarrow}_F |p' $ and $(p',p')\in {\mathcal R}$.
\end{proof}

The next two propositions state the properties of the interaction of general external choice and parallel operator, which are analogous to the expansion law in usual process calculi, e.g., \cite{Milner89}.

\begin{proposition}\label{L:MULTIPLE_I}
Let $n \geq 0$, $m \geq 0$, $A \subseteq Act$ and
$a_i,b_j\in Act$ for each $i<n$ and $j<m$.Then
\[\underset{i< n}{\square}a_i.p_i \parallel_A \underset{j< m}{\square}b_j.q_j
  \sqsubseteq_{RS} ((\square \Omega_1) \Box (\square \Omega_2)) \Box (\square \Omega_3),\]
where
$\Omega_1 = \{a_i.(p_i \parallel_A \underset{j< m}{\square}b_j.q_j)|i<n\;\text{and}\;a_i \notin A \}$,
$\Omega_2 =  \{b_j.(\underset{i< n}{\square}a_i.p_i \parallel_A q_j)|j<m\;\text{and}\;b_j \notin A\}$ and
$\Omega_3 =  \{a_i.(p_i \parallel_A q_j)|a_i = b_j \in A,i<n\;\text{and}\;j<m\}$.
\end{proposition}
\begin{proof}
  Set $N \triangleq \underset{i< n}{\square}a_i.p_i \parallel_A \underset{j< m}{\square}b_j.q_j$ and $M\triangleq  ((\square \Omega_1) \Box (\square \Omega_2)) \Box (\square \Omega_3)$.
 Clearly, both $N$ and $M$ are stable.
  It is sufficient to prove $N\underset{\thicksim}{\sqsubset}_{RS} M$. Put
  \[{\mathcal R}\triangleq\{(N,M)\}\cup Id.\]
  We intend to check that the pair $(N, M)$ satisfies (RS1-4). For (RS1,4), it is straightforward and omitted.

  \textbf{(RS2)} Suppose $M\in F $.
  Then $t \in F $ for some $t \in \Omega_1 \cup \Omega_2 \cup \Omega_3$ by Lemma~\ref{L:BIG_SQUARE}.
  We shall consider the case where $t \in \Omega_1$, the others may be treated similarly and omitted.
  In such case, we may assume that $t \equiv a_{i_0}.(p_{i_0} \parallel_A \underset{j< m}{\square}b_j.q_j)$ with $i_0<n$ and $a_{i_0}\notin A$.
  So $p_{i_0} \in F $ or $\underset{j< m}{\square}b_j.q_j  \in F $.
   Clearly each of them implies   $N\equiv \underset{i< n}{\square}a_i.p_i \parallel_A \underset{j< m}{\square}b_j.q_j \in F $, as desired.

  \textbf{(RS3)} Suppose $N \stackrel{a}{\Rightarrow}_F|p'$.
  Then $ M \notin F $ by $N \notin F $ and (RS2).
  Since $N$ is stable, $N \stackrel{a}{\rightarrow}_F p'' \stackrel{\epsilon}{\Rightarrow}_F|p'$ for some $p''$.
  The proof proceeds by case analysis on the last rule applied in the proof tree of $Strip(\text{CLL}_R,M_{\text{CLL}_R}) \vdash N \stackrel{a}{\rightarrow}p''$.

\noindent  Case 1. $\frac{\underset{i< n}{\square}a_i.p_i \stackrel{a}{\rightarrow} r}{\underset{i< n}{\square}a_i.p_i \parallel_A \underset{j< m}{\square}b_j.q_j \stackrel{a}{\rightarrow}r\parallel_A \underset{j< m}{\square}b_j.q_j}$ with $ \underset{j< m}{\square}b_j.q_j \not\stackrel{\tau}{\rightarrow}$ and $a\notin A$.

  Then $\underset{i< n}{\square}a_i.p_i \stackrel{a}{\rightarrow} r$ and $p'' \equiv r\parallel_A \underset{j< m}{\square}b_j.q_j$.
  By Lemma~\ref{L:BIG_SQUARE}(3), we have $a= a_{i_0}$ and $r \equiv p_{i_0}$ for some $i_0 < n$.
  Due to $a_{i_0}=a \notin A$, $a_{i_0}.(p_{i_0} \parallel_A \underset{j< m}{\square}b_j.q_j) \in \Omega_1$.
  So $\square \Omega_1 \stackrel{a_{i_0}}{\rightarrow}p_{i_0} \parallel_A \underset{j< m}{\square}b_j.q_j$ by Lemma~\ref{L:BIG_SQUARE}(2).
  Moreover, since $\{a_i,b_j|i<n\;\text{and}\;j<m\}\subseteq Act$, we get $\square \Omega_2 \not\stackrel{\tau}{\rightarrow}$ and $\square \Omega_3 \not\stackrel{\tau}{\rightarrow}$ by Lemma~\ref{L:BIG_SQUARE}(3).
  Then $M \stackrel{a_{i_0}}{\rightarrow}p_{i_0} \parallel_A \underset{j< m}{\square}b_j.q_j\equiv p''$.
  Hence, $ M \stackrel{a}{\Rightarrow}_F|p'$ and $(p',p')\in {\mathcal R}$.\\

\noindent  Case 2. $\frac{\underset{j< m}{\square}b_j.q_j \stackrel{a}{\rightarrow} r}{\underset{i< n}{\square}a_i.p_i \parallel_A \underset{j< m}{\square}b_j.q_j \stackrel{a}{\rightarrow} \underset{i< n}{\square}a_i.p_i \parallel_A r}$ with $\underset{i< n}{\square}a_i.p_i \not\stackrel{\tau}{\rightarrow}$ and $a\notin A$.

  Similar to Case 1.\\

\noindent  Case 3. $\frac{\underset{i< n}{\square}a_i.p_i \stackrel{a}{\rightarrow} r,\underset{j< m}{\square}b_j.q_j \stackrel{a}{\rightarrow} s}{\underset{i< n}{\square}a_i.p_i \parallel_A \underset{j< m}{\square}b_j.q_j \stackrel{a}{\rightarrow}r\parallel_A s}$ with $a\in A$.

  Then $\underset{i< n}{\square}a_i.p_i \stackrel{a}{\rightarrow} r$, $ \underset{j< m}{\square}b_j.q_j \stackrel{a}{\rightarrow} s $ and $p'' \equiv r\parallel_A s$.
  By Lemma~\ref{L:BIG_SQUARE}(3), we have $a= a_{i_0}$, $r \equiv p_{i_0}$ for some $i_0 < n$ and $a = b_{j_0}$, $s \equiv q_{j_0}$ for some $j_0 < m$.
  Then $a_{i_0}.(p_{i_0} \parallel_A q_{j_0}) \in \Omega_3$ because of $a_{i_0} = b_{j_0}=a \in A$.
  So $\square \Omega_3 \stackrel{a_{i_0}}{\rightarrow}p_{i_0} \parallel_A q_{j_0}$ by Lemma~\ref{L:BIG_SQUARE}(2).
  Moreover, since $\{a_i,b_j|i<n\;\text{and}\;j<m\}\subseteq Act$, we get $\square \Omega_1 \not\stackrel{\tau}{\rightarrow}$ and $\square \Omega_2 \not\stackrel{\tau}{\rightarrow}$ by Lemma~\ref{L:BIG_SQUARE}(3). Then $M \stackrel{a_{i_0}}{\rightarrow}p_{i_0} \parallel_A q_{j_0}\equiv p''$.
  Hence, $ M \stackrel{a}{\Rightarrow}_F|p'$ and $(p',p')\in {\mathcal R}$.
\end{proof}

Compared with usual expansion law in process calculus, e.g., Prop. 3.3.5 in \cite{Milner89}, someone may expect that the inequation below holds, where $\Omega_i$ ($1\leq i\leq 3$) is same as ones in Prop.~\ref{L:MULTIPLE_I}.
\[((\square \Omega_1) \Box (\square \Omega_2)) \Box (\square \Omega_3) \sqsubseteq_{RS}    \underset{i< n}{\square}a_i.t_i \parallel_A \underset{j< m}{\square}b_j.s_j. \tag{EXP}\]
Unfortunately, it isn't valid.
For instance, consider $a_0.t_0 \triangleq a.\bot$, $a_1.t_1 \triangleq c.0$ and $b_0.s_0 \triangleq b.0$ with $a\neq b\neq c$.
Let $A = \{a,b\}$.
Clearly, the set $\Omega_i(1 \leq i \leq 3)$ corresponding to ones in the above proposition are: $\Omega_1 =\{c.(0\parallel_{\{a,b\}}b.0)\}$ and $\Omega_2 = \Omega_3 = \emptyset$.
Then \[(( \square \Omega_1) \Box (\square \Omega_2)) \Box ( \square \Omega_3 ) \equiv (c.(0 \parallel_{\{a,b\}} b.0)\Box 0)\Box 0.\]
By Lemma~\ref{L:F_NORMAL}, $(a.\bot \Box c.0) \parallel_{\{a,b\}} b.0 \in F$ and  $(c.(0\parallel_{\{a,b\}}b.0) \Box 0)\Box 0 \notin F$.
Then it is easy to see that $(c.(0\parallel_{\{a,b\}}b.0) \Box 0)\Box 0 \not \sqsubseteq_{RS} (a.\bot \Box c.0) \parallel_{\{a,b\}} b.0 $.

%Consider $(a.\bot \Box c.0) \parallel_{\{a,b\}} b.0$ and  $(c.(0\parallel_{\{a,b\}}b.0) \Box 0)\Box 0 $ with $a \neq b\neq c$.
%By Lemma~\ref{L:F_NORMAL}, $(a.\bot \Box c.0) \parallel_{\{a,b\}} b.0 \in F $ and  $(c.(0\parallel_{\{a,b\}}b.0) \Box 0)\Box 0 \notin F $.
%Then $(c.(0\parallel_{\{a,b\}}b.0) \Box 0)\Box 0 \not \sqsubseteq_{RS} (a.\bot \Box c.0) \parallel_{\{a,b\}} b.0 $.
However, the inequation (EXP) holds for processes satisfying a moderate condition. Formally, we have the result below.

\begin{proposition}\label{L:MULTIPLE_IV}
Let $n,m\geq0$, $A \subseteq Act$ and $a_i,b_j\in Act$ for each $i<n$ and $j<m$.
Assume that $(\{p_i|a_i\in A\;\text{and}\; a_i\neq b_j\;\text{for each}\; j<m\}\cup \{q_j|b_j\in A\;\text{and}\; b_j\neq a_i\;\text{for each}\;i<n\}) \cap F  =\emptyset$,
 then \[((\square \Omega_1) \Box (\square \Omega_2)) \Box (\square \Omega_3) \sqsubseteq_{RS}    \underset{i< n}{\square}a_i.p_i \parallel_A \underset{j< m}{\square}b_j.q_j\]
where $\Omega_i$ ($1\leq i\leq 3$) is same as ones in Prop.~\ref{L:MULTIPLE_I}.
\end{proposition}
\begin{proof}
  Set $M \triangleq \underset{i< n}{\square}a_i.p_i \parallel_A \underset{j< m}{\square}b_j.q_j$ and $N \triangleq  ((\square \Omega_1) \Box (\square \Omega_2)) \Box (\square \Omega_3)$.
  Similar to Prop.~\ref{L:MULTIPLE_I}, we shall prove $N\underset{\thicksim}{\sqsubset}_{RS} M$. Put
  ${\mathcal R}\triangleq\{(N,M)\}\cup Id$.
  It suffices to show that $\mathcal R$ is a stable ready simulation relation. We will check that the pair $(N, M)$ satisfies (RS2), the remainder is analogous to ones of Prop.~\ref{L:MULTIPLE_I}.

  \textbf{(RS2)} Suppose $M \in F $. By Lemmas~\ref{L:F_NORMAL} and \ref{L:BIG_SQUARE}, we get either $p_{i_0} \in F$ for some $i_0<n$ or $q_{j_0} \in F$  for some $j_0<m$.
  W.l.o.g, we consider the first alternative.
  Then, by the assumption,  $a_{i_0} \notin A$ or $a_{i_0} = b_{j_0}$ for some $j_0 < m$.
  Consequently,  $a_{i_0}.(p_{i_0} \parallel_A \underset{j< m}{\square}b_j.q_j) \in \Omega_1$ or $a_{i_0}.(p_{i_0} \parallel_A q_{j_0}) \in \Omega_3$. Hence $N\in F $ by Lemma~\ref{L:BIG_SQUARE}, as desired.
\end{proof}

%Taking $m=0$ in Prop.~\ref{L:MULTIPLE_I} and \ref{L:MULTIPLE_IV}, we get

%\begin{proposition}\label{L:MULTIPLE_III}
%Let $n>0$, $A \subseteq Act$ and
%$a_i\in Act$ for each $i<n$.
%\begin{enumerate}
%  \item $\underset{i< n}{\square}a_i.t_i \parallel_A 0 \sqsubseteq_{RS} \underset{i< n,a_i\notin A}{\square}a_i.(t_i \parallel_A 0)$.
%  \item If $\{t_i|a_i\in A\;\text{and}\;i<n\}\cap F =\emptyset$ then $ \underset{i< n,a_i\notin A}{\square}a_i.(t_i \parallel_A 0) \sqsubseteq_{RS}    \underset{i< n}{\square}a_i.t_i \parallel_A 0$.
%\end{enumerate}
%\end{proposition}
%\begin{proof}
%Similar to Prop.~\ref{L:MULTIPLE_I} and \ref{L:MULTIPLE_IV}.
%\end{proof}

We now have all of the properties that we need to prove the soundness of the axiomatic system $AX_{\text{CLL}}$.

\begin{theorem}[Soundness]\label{T:SOUNDNESS}
If $\vdash p\leqslant q$ then $p\sqsubseteq_{RS} q$ for any $p,q \in T(\Sigma_{\text{CLL}_R})$.
\end{theorem}
\begin{proof}%
%As usual, it is sufficient to show that (1) for each instance $p\leqslant q$ of every axiom, $p\sqsubseteq_{RS}q$   and
%  (2) all inference rules preserve $\sqsubseteq_{RS}$.
Immediately follows from Lemmas~\ref{L:CON_ID_I}(1) and \ref{L:BPT}, Prop.~\ref{L:ONE_OPERATOR}, \ref{S:DISTRIBUTIVE}, \ref{S:SPECIAL}, \ref{L:MULTIPLE_VI}, \ref{L:MULTIPLE_II}, \ref{L:MULTIPLE_I} and \ref{L:MULTIPLE_IV},  Theorem~\ref{L:precongruence} and the fact that $\sqsubseteq_{RS}$ is reflexive and transitive.
\end{proof}

\section{Normal form and ground-completeness}

This section will establish the ground-completeness of $AX_{\text{CLL}}$ for processes that  are generated by BNF
\[ t::= 0\;|\perp\;|\;(\alpha.t) \;|\; (t\Box t)\;|\;(t\wedge t)\;|\;(t\vee t)\;|\;(t\parallel_A t). \]
The set of all these processes is denoted by $T(\Sigma_{\text{CLL}})$.

To prove the ground-completeness of $AX_{\text{CLL}}$, we use a standard technique involving normal forms.
The idea is to isolate a particular subclass of terms, called normal forms, such that the proof of the completeness is straightforward for it.
The completeness for arbitrary terms will follow if we can show that each term can be reduced to normal form using axioms and inference rules in $AX_{\text{CLL}}$.
Therefore the proof of ground-completeness falls naturally into two parts: first, we will show that each process in $T(\Sigma_{\text{CLL}})$ is normalizable; second, it will be demonstrated that $AX_{\text{CLL}}$ is ground-complete w.r.t processes in normal form.
Before defining the normal form, we first introduce two useful notations.\\

\noindent \textbf{Notation}
\begin{enumerate}
  \item $Prefix(\underset{i<n}{\square}a_i.t_i)\triangleq \{a_i|i<n\}$.
  \item Let $<t_0,\dots,t_{n-1}>$ be a finite sequence of process terms with $n > 0$.
The general disjunction $\underset{i<n}\bigvee t_i$ is defined as
    \[\underset{i<1}\bigvee t_i \triangleq t_0,\;\text{and}\;
    \underset{i<k+1}\bigvee t_i \triangleq (\underset{i<k}\bigvee t_i) \vee t_k\; \text{for} \;k\geq 1.\]
    Similar to general external choice, the order and grouping of terms in $\underset{i<n}\bigvee t_i$  may be ignored by virtue of Axioms $DI1$ and $DI2$.
\end{enumerate}

\begin{mydefn}[Normal Form]\label{D:NORMAL_FORM}
    The set $NF_B$ is the least subset of $T(\Sigma_{\text{CLL}})$ such that $\underset{i< n}{\bigvee}t_i\in NF_B$ if $n > 0$ and for each $i<n$, $t_i$ has the format $\underset{j< m_i}{\square}a_{ij}.t_{ij}$ with $m_i \geq 0$ such that

        \noindent (N)\;\;\;  $t_{ij}\in NF_B$ for each $j<m_i$,

        \noindent (D) \;\; $\underset{j< m_i}{\square}a_{ij}.t_{ij}$ is injective in prefixes, and

        \noindent (N-$\tau$)  $a_{ij} \in Act$ for each $j<m_i$.

We put $NF \triangleq\{ \bot \} \cup NF_B$.
Each process term in $NF$ is  in normal form.
Notice that $NF_B \subseteq T(\Sigma_B)$, and $0\in NF_B$ by taking $n=1$ and $m_0 = 0$ in $ \underset{i<n}{\bigvee}\underset{j<m_i}{\square}t_{ij}$.
\end{mydefn}

The following simple observations inspire the format of normal processes in $NF_B$.

First, due to $\tau$-purity, the behaviour of any process consists of external and internal choices, which are interleaving but never mixing.
This fact induces us to adopt the format $\underset{i< n}{\bigvee}\underset{j<m_i}{\square}t_{ij}$ as normal forms.

Second, because of $a.p \Box a.q =_{RS} a.(p \vee q)$ for $p,q \in T(\Sigma_B)$ and $\tau.p=_{RS}p$, we may require normal forms to satisfy Conditions~(D) and (N-$\tau$), which make demonstrating the completeness w.r.t $NF$ (see Lemma~\ref{L:COMPLETENESS}) easier.
In fact, processes $\underset{j< m_i}{\square}a_{ij}.t_{ij}$ satisfying (N-$\tau$) indeed are $\vee$-irreducible in the distributive lattice $T(\Sigma_{\text{CLL}},\vee,\wedge)$ (see Remark~\ref{R:NORMAL_FORM} given later).
Hence, from the lattice-theoretical viewpoint, defining normal form as above is natural.

In the following, we will show that each process term can be transformed using axioms in $AX_{\text{CLL}}$ into a normal form.
To this end, the next four lemmas are firstly proved.

\begin{lemma}\label{L:DIS_INEQUATION}
\noindent (1) $\vdash a.t \Box a.s  \leqslant a.(t \vee s) $.

\noindent (2) $\vdash (t\odot s_1)\vee(t\odot s_2)\leqslant t \odot (s_1 \vee s_2)$ for each $\odot \in \{\Box, \wedge, \parallel_A\}$.

\end{lemma}
\begin{proof}

\noindent \textbf{(1)} $ \vdash t  \leqslant t \vee s$ and $\vdash s  \leqslant t \vee s$   \qquad\;\;\;\;\; (by  $DI1$ ,  $DI5$ and TRANS )

\noindent $\Rightarrow   \vdash a.t  \leqslant a.(t \vee s)$ and $\vdash a.s  \leqslant a.(t \vee s)$  \;\qquad\qquad\qquad\;  (by CONTEXT)

\noindent $\Rightarrow   \vdash a.t \Box a.s \leqslant a.(t \vee s)$ \qquad\qquad\qquad (by CONTEXT,  $EC3$ and TRANS)

\noindent \textbf{(2)}  $\vdash s_1 \leqslant s_1 \vee s_2$ and $\vdash s_2 \leqslant s_1 \vee s_2$  \qquad(by $DI1$, $DI5$ and TRANS)

\noindent   $\Rightarrow \vdash t \odot s_1 \leqslant t \odot (s_1 \vee s_2)$ and $\vdash t\odot s_2 \leqslant t\odot (s_1 \vee s_2)$ (by CONTEXT and REF)

\noindent  $\Rightarrow \vdash (t \odot s_1) \vee (t\odot s_2) \leqslant t \odot (s_1 \vee s_2)$  \qquad (by $DI3$, CONTEXT and TRANS)
%\end{proof}
%
%\begin{lemma}\label{L:SP3}
%      $\vdash a.t \Box a.s  \leqslant a.(t \vee s) $.
%\end{lemma}
%\begin{proof}
\end{proof}

The next three lemmas provide a series of closure properties of $NF$, which ensure that the inductive proof of Normal Form Theorem can be carried out smoothly.

\begin{lemma}\label{L:COMP_CONJ}
  If $t,s\in NF_B$ then $\vdash t \wedge s = r$ for some $r \in NF$.
\end{lemma}
\begin{proof}
We prove it by induction on the number $|t|+|s|$ \footnote{$|t|$ is the number of operators occurring in $t$.}.
  Since $t,s \in NF_B$, we may assume that $t \equiv \underset{i< n}{\bigvee}t_{i}$ and $s \equiv \underset{i'< n'}{\bigvee}s_{i'}$.
   By $DI1$, $DI2$, $CO1$, $DS2$ and Lemma~\ref{L:DIS_INEQUATION}(2), we get
   \[ \vdash t \wedge s
    = \bigvee_{i< n,i'< n'}(t_{i}\wedge s_{i'}) \tag{\ref{L:COMP_CONJ}.1}.\]
    Let $i< n$ and $i'< n'$. We will show that $\vdash t_{i} \wedge s_{i'}=r_{ii'}$ for some $r_{ii'} \in NF$. Clearly, we may assume that $t_{i}\equiv \underset{j< m_{i}}{\square}a_{ij}.t_{ij}$ and $s_{i'}\equiv \underset{j'< m_{i'}'}{\square}b_{i'j'}.s_{i'j'}$ satisfying (N), (D) and (N-$\tau$) in Def.~\ref{D:NORMAL_FORM}. We consider two cases below.\\

    \noindent Case 1. $Prefix(t_{i})\neq Prefix(s_{i'})$.

          By  $ECC1$, we have $\vdash t_{i} \wedge s_{i'}  = \bot $.\\

    \noindent Case 2.  $Prefix(t_{i})= Prefix(s_{i'})$.

        Thus, by the item (D) in Def.~\ref{D:NORMAL_FORM}, we have $m_i = m_{i'}'$.
        If $m_i=0$ then, by the definition of general external choice, we get $t_i\equiv s_{i'} \equiv 0$.
        Moreover, $\vdash t_i \wedge s_{i'}=0$ follows from $CO2$.
        In the following, we consider the nontrivial case where $m_i>0$.
        By $EC1$, $EC2$, $ECC2$ and $ECC3$, it follows that
        \[\vdash t_{i} \wedge s_{i'}  = \underset{\begin{subarray}
                   \; j,j'<m_i,\\
                   a_{ij}= b_{i'j'}
                \end{subarray}}\square a_{ij}.(t_{ij} \wedge s_{i'j'}). \]
         For each pair $j,j'< m_i$  with $a_{ij}= b_{i'j'}$,
         since $t_{ij},s_{i'j'} \in NF_B$ and $|t|+|s|>|t_{ij}|+|s_{i'j'}|$, by IH, we have $\vdash t_{ij} \wedge s_{i'j'} = t_{iji'j'}$ for some  $t_{iji'j'} \in NF$.
         Set \[S \triangleq \underset{\begin{subarray}
                   \; j,j'<m_i,\\
                   a_{ij}= b_{i'j'}
                \end{subarray}}\square a_{ij}.t_{iji'j'}.\]
        Consequently, by  CONTEXT and TRANS, we have
            \[\vdash t_{i} \wedge s_{i'} = S.\]
        Clearly, if $t_{iji'j'}\in NF_{B}$ for each pair $j,j'<m_i$ with  $a_{ij}= b_{i'j'}$, then $S \in NF_{B}$.
        Otherwise, we have $t_{ij_0i'j_0'} \equiv \bot$ for some $j_0,j_0'<m_i$, then it follows from $PR1$ that \[\vdash  a_{ij_0}.t_{ij_0i'j_0'}  =\bot.\]
        Further, by $EC5$, CONTEXT and TRANS, we get $\vdash S =\bot$.
%        Hence, $\vdash t_{i} \wedge s_{i'}= \bot$.

     In summary, it follows from the discussion above that, for each $i<n$ and $i'<n'$,
     \[\text{either} \vdash t_{i} \wedge s_{i'}= r_{ii'}\;\text{for some}\;r_{ii'}\in NF_{B}\;\text{or}\; \vdash t_{i} \wedge s_{i'}= \bot.\]
     Then, by $DI1$, $DI4$ and (\ref{L:COMP_CONJ}.1), $\vdash t \wedge s = r$ for some $r \in NF_{B}$ or $\vdash t \wedge s= \bot$.
\end{proof}

In the above proof, we do not  explicitly show the proof for the induction basis where $t \equiv s \equiv 0$, as it is an instance of the proof of the induction step.

\begin{lemma}\label{L:BIG_SQUARE_EC}
  If $t \equiv \underset{i<n}{\square}a_i.t_i \in NF_B$  and $s \equiv \underset{j<m}{\square}b_j.s_j \in NF_B$, then $\vdash t \Box s = \underset{i<k}\square c_i.r_i$ for some $\underset{i<k}\square c_i.r_i \in NF_B$.
\end{lemma}
\begin{proof}
If $n=0$ or $m=0$ then it immediately follows from $EC1$ and $EC4$ due to the definition of general external choice.
In the following, we consider the non-trivial case where $n>0$ and $m>0$.
We distinguish two cases below.\\

\noindent Case 1. $Prefix(t) \cap Prefix(s)=\emptyset$.

          Set
            \[p_k \triangleq \begin{cases}
                a_k.t_k  & k<n,\\
                b_{k-n}.s_{k-n} &  n \leq k < m+n.\\
                \end{cases}\]
        Then, it is trivial to check that $\underset{k<m+n}\square p_k $ satisfies (N), (D) and (N-$\tau$) in Def.~\ref{D:NORMAL_FORM}, that is, $\underset{k<m+n}\square p_k \in NF_B$.
        Moreover, by $EC2$ and TRANS, it immediately follows that $\vdash t \Box s = \underset{k<m+n}\square p_k$.\\

\noindent Case 2. $Prefix(t) \cap Prefix(s) \neq \emptyset $.

  Let $i_0<n$ and $j_0<m$ with $a_{i_0} = b_{j_0}$, since $NF_B \subseteq T(\Sigma_B)$, by Lemma~\ref{L:DIS_INEQUATION}(1) and $DS4$, we get $\vdash a_{i_0}.t_{i_0} \Box b_{j_0}.s_{j_0} = a_{i_0}.(t_{i_0} \vee s_{j_0})$.
  Further, by Def.~\ref{D:NORMAL_FORM}, $DI1$, $DI2$, CONTEXT and TRANS, it follows from $t_{i_0},s_{j_0}\in NF_B$ that
  \[\vdash a_{i_0}.t_{i_0} \Box b_{j_0}.s_{j_0} = a_{i_0}.p\;\text{for some}\;p\in NF_B.\]
  Thus, for each $i<n$ and $j<m$ with $a_i = b_j$, we can fix a process term $p_{ij}\in NF_B$ such that
  \[\vdash a_i.t_i \Box b_j.s_j = a_i.p_{ij}.\]
  Put

\[
     S_1 \triangleq \underset{ \begin{subarray} \;a_i \notin Prefix(s),\\\;\;\;\;\;\;\;i<n\end{subarray}}{\square}a_i.t_i,  \;
       S_2 \triangleq \underset{ \begin{subarray} \;b_j \notin Prefix(t),\\\;\;\;\;\;\;\;j<m\end{subarray}}{\square}b_j.s_j, \;
       S_3 \triangleq \underset{\begin{subarray} \;a_i \in Prefix(t)\cap Prefix(s),\\\;\;\;\;\;a_i = b_j,i<n,j<m\end{subarray}
       }{\square}a_i.p_{ij}.\]
 % \begin{enumerate}
%    \item $S_1 \triangleq \underset{ \begin{subarray} \;a_i \notin Prefix(s),\\\;\;\;\;\;\;\;i<n\end{subarray}}{\square}a_i.t_i$,
%    \item $S_2 \triangleq \underset{ \begin{subarray} \;b_j \notin Prefix(t),\\\;\;\;\;\;\;\;j<m\end{subarray}}{\square}b_j.s_j$,
%    \item $S_3 \triangleq \underset{\begin{subarray} \;a_i \in Prefix(t)\cap Prefix(s),\\\;\;\;\;\;a_i = b_j,i<n,j<m\end{subarray}
%       }{\square}a_i.p_{ij}$.
%  \end{enumerate}
      % \[S_1 \triangleq \underset{ \begin{subarray} \;a_i \notin Prefix(s),\\\;\;\;\;\;\;\;i<n\end{subarray}}{\square}a_i.t_i,\;
%       S_2 \triangleq \underset{ \begin{subarray} \;b_j \notin Prefix(t),\\\;\;\;\;\;\;\;j<m\end{subarray}}{\square}b_j.s_j, \;
%       S_3 \triangleq \underset{\begin{subarray} \;a_i \in Prefix(t)\cap Prefix(s),\\\;\;\;\;\;a_i = b_j,i<n,j<m\end{subarray}
%       }{\square}a_i.p_{ij}.\]

    \noindent Then, by $EC1$, $EC2$, TRANS and CONTEXT, we obtain $\vdash t \Box s = (S_1 \Box S_2)\Box S_3$.
    Clearly, both $S_1$ and $S_2$ are in $NF_B$.
    Moreover, since $t$ and $s$ are injective in prefixes, so is $S_3$.
    Hence, $S_3$ is also in $NF_B$.
    Further, since $Prefix(S_i)\cap Prefix(S_j) = \emptyset$ for $1 \leq i \neq j \leq 3$, similar to Case 1, we have
    $\vdash (S_1 \Box S_2)\Box S_3 = \underset{i<k}\square c_i.r_i$ for some $\underset{i<k}\square c_i.r_i \in NF_B$.
\end{proof}

\begin{lemma}\label{L:COMP_PARALLEL}
  If $t,s \in NF_B$ then $\vdash  t\parallel_A s  = r$ for some $r \in NF_B$.
\end{lemma}
\begin{proof}
We prove it by induction on the number $|t|+|s|$.
Since $t,s \in NF_B$, we may assume that $t\equiv \underset{i< n}{\bigvee}t_{i}$ and $s\equiv \underset{i'< n'}{\bigvee}s_{i'}$.
 By axioms $DI1$, $DI2$, $PA1$, $DS3$ and Lemma~\ref{L:DIS_INEQUATION}(2), we get
    \[\vdash  t \parallel_A s
    = \underset{i< n,i'< n'}{\bigvee}(t_{i}\parallel_A s_{i'}). \tag{\ref{L:COMP_PARALLEL}.1}\]
    We shall show that for each $i< n$ and $i'< n'$,
    \[\vdash t_{i}\parallel_A s_{i'} =r_{ii'}\;\text{for some}\;r_{ii'} \in NF_B.\]
     Let $i< n$ and $i'< n'$.
     We may assume that $t_{i}\equiv \underset{j< m_{i}}{\square}a_{ij}.t_{ij}$ and $s_{i'}\equiv \underset{j'< m_{i'}'}{\square}b_{i'j'}.s_{i'j'}$ satisfying (N), (D) and (N-$\tau$) in Def.~\ref{D:NORMAL_FORM}.
     By $EXP1$ and $EXP2$, we have
            \begin{multline*}
                \vdash t_{i}\parallel_A s_{i'} =\\
                (\underset {\begin{subarray}
                   \;j< m_i,\\
                   a_{ij} \notin A
                \end{subarray}}
                \square a_{ij}.(t_{ij} \parallel_A s_{i'}) \Box
                \underset {\begin{subarray}
                   \;j'< m_{i'}',\\
                   b_{i'j'} \notin A
                \end{subarray}}
                \square
                b_{i'j'}.(t_{i} \parallel_A s_{i'j'})) \Box
             \underset {\begin{subarray}
                   \;j< m_i,j'<m_{i'}',\\
                  a_{ij}= b_{i'j'}\in A
                \end{subarray}}
                \square
                a_{ij}.(t_{ij} \parallel_A s_{i'j'}). \tag{\ref{L:COMP_PARALLEL}.2}
            \end{multline*}
    We consider two cases.\\

\noindent    Case 1. $m_i=0$ or $m_{i'}'=0$.

    W.l.o.g, assume that $m_i=0$. Then, by (\ref{L:COMP_PARALLEL}.2), $EC1$, $EC4$, CONTEXT and TRANS, we get
    \[ \vdash t_i \parallel_A s_{i'} = \underset {\begin{subarray}
                   \;j'< m_{i'}',\\
                   b_{i'j'} \notin A
                \end{subarray}}
                \square
                b_{i'j'}.( 0 \parallel_A s_{i'j'}). \tag{\ref{L:COMP_PARALLEL}.3}\]
    If $\{b_{i'j'} \notin A|j'<m_{i'}'\}=\emptyset$ then $\vdash t_i \parallel_A s_{i'} = 0$.
    Next, we consider the case where $\{b_{i'j'}\notin A|j'<m_{i'}'\} \not= \emptyset$.
    For each $j' < m_{i'}'$ with $b_{i'j'}\notin A$, we have $s_{i'j'}\in NF_B$, moreover, $|t|+|s|>|0|+|s_{i'j'}|$.
    Then, by IH, we get $\vdash 0\parallel_A s_{i'j'} = p_{j'}$ for some $p_{j'} \in NF_B$.
    Therefore, by CONTEXT, TRANS and (\ref{L:COMP_PARALLEL}.3), it is easy to see that $\vdash t_i \parallel_A s_{i'} =  r_{ii'}$ for some $r_{ii'} \in NF_B$.\\

\noindent    Case 2. $m_i > 0$ and  $m_{i'}' > 0$.

    In such case, for each $j<m_i$ and $j'<m_{i'}'$, we have $|t|+|s|>|t_{ij}|+|s_{i'}|$, $|t|+|s|>|t_{i}|+|s_{i'j'}|$ and $|t|+|s|>|t_{ij}|+|s_{i'j'}|$.
      Moreover, $t_{ij},s_{i'},t_{i},s_{i'j'}\in NF_B$.
      Then, by IH, there exist $t_{iji'},t_{ii'j'},t_{iji'j'}\in NF_B$ such that $ \vdash  t_{ij} \parallel_A s_{i'} = t_{iji'}$, $\vdash t_{i} \parallel_A s_{i'j'} = t_{ii'j'}$  and $ \vdash  t_{ij} \parallel_A s_{i'j'} =t_{iji'j'}$.
            Set
\[
       S_1 \triangleq \underset{ \begin{subarray} \;a_{ij} \notin A,\\j<m_i\end{subarray}}{\square}a_{ij}.t_{iji'}, \;
        S_2 \triangleq \underset{ \begin{subarray} \;b_{i'j'} \notin A,\\j'<m_{i'}'\end{subarray}}{\square}b_{i'j'}.t_{ii'j'}, \;
       S_3 \triangleq \underset{ \begin{subarray} \;a_{ij} = b_{i'j'}\in A,\\j'<m_{i'}',j<m_i\end{subarray}}{\square}a_{ij}.t_{iji'j'}.\]

    Clearly, $S_1,S_2,S_3\in NF_B$ and $\vdash t_{i}\parallel_A s_{i'}=(S_1 \Box S_2) \Box S_3$.
    Further, by Lemma~\ref{L:BIG_SQUARE_EC}, we get $\vdash t_{i}\parallel_A s_{i'}=r_{ii'}$ for some $r_{ii'} \in NF_B$, as desired.

     In summary, by the discussion above, we conclude that, for each $i<n$ and $i'<n'$, $\vdash t_{i}\parallel_A s_{i'} =r_{ii'}$ for some $r_{ii'}\in NF_B$.
     Then, by Def.~\ref{D:NORMAL_FORM} and (\ref{L:COMP_PARALLEL}.1), it immediately follows that $\vdash t \parallel_A s = r $ for some $r \in NF_B$, as desired.
\end{proof}

Now, we can prove that each process term is normalizable. That is

\begin{theorem}[Normal Form Theorem]\label{T:NORMALFORM}
  For each $t \in T(\Sigma_{\text{CLL}})$, $\vdash  t = s$   for some $s \in NF$.
\end{theorem}
\begin{proof}
  We prove it by induction on the structure of  $t$.

\noindent $\bullet$ $t \equiv 0$ or $t \equiv \bot$.

    Trivially.

\noindent $\bullet$ $t \equiv \alpha.t_1$.

        By IH and CONTEXT, we get $\vdash t = \alpha.t_1'$ for some $t_1' \in NF$.
        If $t_1'\not\equiv \bot$ and $\alpha\in Act$, then $\alpha.t_1'\in NF_B$.
        If $t_1' \equiv \bot$, by  $PR1$, $PR2$ and TRANS, we obtain $\vdash t = \bot$.
        If $\alpha = \tau$, by $PR2$ and TRANS, we have  $\vdash t = t_1'$.

\noindent $\bullet$ $t \equiv t_1 \odot t_2$ with $\odot \in \{\vee,\Box,\wedge,\parallel_A\}$.

        For $i=1,2$, by IH, we have $\vdash t_i = t_i'$ for some $t_i'\in NF$. We distinguish four cases based on $\odot$.\\

\noindent Case 1. $\odot = \vee$.

        If  $t_1' \not\equiv \bot$  and  $t_2' \not\equiv \bot$ (i.e., $t_1',t_2'\in NF_B$), then it immediately follows from $DI1$, $DI2$, CONTEXT and TRANS that $\vdash t = s$ for some $s \in NF_B$.
        Otherwise, w.l.o.g, assume that $t_1' \equiv \bot$.
        Then, by $DI1$, $DI4$ and TRANS, we get $\vdash  t = t_2'$.\\

\noindent Case 2. $\odot = \Box$.

        If either $t_1' \equiv \bot$ or $t_2' \equiv \bot$,  then it follows from $EC1$ and $EC5$ that $\vdash  t = \bot$.
        In the following, we consider the case where $t_1' \not\equiv \bot$ and $t_2' \not\equiv \bot$.
        In this situation, we get $t_1',t_2'\in NF_B$.
        So, we may assume that $t_1' \equiv \underset{i< n}{\bigvee}\underset{j<m_i}\square a_{ij}.s_{ij}$ and $t_2' \equiv \underset{i'< n'}{\bigvee}\underset{j'<m_{i'}'}\square b_{i'j'}.r_{i'j'}$ with $\underset{j<m_i}\square a_{ij}.s_{ij},\underset{j'<m_{i'}'}\square b_{i'j'}.r_{i'j'}\in NF_B$ for each $i<n$ and $i'<n'$.
        Thus, by $DI1$, $DI2$, CONTEXT, TRANS, $DS1$ and Lemma~\ref{L:DIS_INEQUATION}(2),  we obtain
        \[\vdash t_1 \Box t_2 =
        \underset{i < n, i' < n'}{\bigvee}(\underset{j<m_i}\square a_{ij}.s_{ij} \Box \underset{j'<m_{i'}'}\square b_{i'j'}.r_{i'j'}).\]
        Further, by CONTEXT, Lemma~\ref{L:BIG_SQUARE_EC} and Def.~\ref{D:NORMAL_FORM}, it immediately follows that $\vdash t_1 \Box t_2 = t_3$ for some $t_3\in NF_B$.\\

\noindent Case 3. $\odot = \wedge$.

         If $t_i'\in NF_B$ for $i=1,2$ then, by Lemma~\ref{L:COMP_CONJ}, we have $\vdash t = t_3$ for some $t_3 \in NF$, otherwise, by $CO1$ and $CO3$, we get $\vdash t = \bot$.\\

\noindent Case 4. $\odot = \parallel_A$.

        If either $t_1' \equiv \bot$ or $t_2' \equiv \bot$ then, by $PA1$ and $PA2$, we get $\vdash t = \bot$.
        Otherwise,  we have $t_1',t_2' \in NF_B$, so, by Lemma~\ref{L:COMP_PARALLEL}, we obtain $\vdash  t= s$ for some $s \in NF_B$.
\end{proof}

\begin{rmk}\label{R:NORMAL_FORM}
  Clearly, $\underset{i<n}{\square}a_i.t_i=_{RS}p \vee q$ with $a_i \in Act$ implies $\underset{i<n}{\square}a_i.t_i=_{RS}p$ or $\underset{i<n}{\square}a_i.t_i=_{RS}q$, and $\bot =_{RS}p \vee q$ implies $\bot=_{RS}p$ and $\bot =_{RS}q$ for any $p,q$.
  Thus $\bot$ and processes with form $\underset{i<n}{\square}a_i.t_i$ are $\vee$-irreducible in the distributive lattice $<T(\Sigma_{\text{CLL}}),\vee,\wedge>$.
  Therefore, by the well-known result so-called Unique Decomposition Theorem in Lattice Theory (see, e.g. \cite{Birkhoff}), the normal form representation of any $t\in T(\Sigma_{\text{CLL}})$ is unique in an obvious sense.
\end{rmk}
We now turn our attention to the ground-completeness of $AX_{\text{CLL}}$.
First, we state a trivial result about general disjunction.

\begin{lemma}\label{L:SUBSTRUCTURE}
  Let $n>0$ and $t_i$ be stable for each $i< n$.

\noindent (1) If $\underset{i< n}{\bigvee}t_i \notin F $ then $\underset{i< n}{\bigvee}t_i \stackrel{\epsilon}{\Rightarrow}_F| t_{i}$ for each $i < n$.

\noindent (2) If $\underset{i< n}{\bigvee}t_i \stackrel{\epsilon}{\Rightarrow} | t'$ then $t'\equiv t_{i_0}$ for some $i_0 < n$.
\end{lemma}
\begin{proof}
  Straightforward by induction on $n$.
\end{proof}

A crucial step in proving the ground-completeness is to verify the completeness of $AX_{\text{CLL}}$ w.r.t $NF$. Next we do this.

\begin{lemma}\label{L:COMPLETENESS}
  If $t_1,t_2 \in NF$ and $t_1 \underset{\thicksim}{\sqsubset}_{RS} t_2$ then $\vdash  t_1 \leqslant t_2$
\end{lemma}
\begin{proof}
    We prove the statement by induction on  $|t_1|$.
    Since $t_1 \underset{\thicksim}{\sqsubset}_{RS} t_2$, both $t_1$ and $t_2$ are stable. Further, since $t_1,t_2\in NF$, we get, for $i=1,2$

    \[t_i \equiv 0\;\text{or}\; t_i \equiv \bot\;\text{or}\; t_i \equiv \underset{j<n_i}{\square}a_{ij}.t_{ij}\in NF_B\;\text{with}\; n_i>0.  \tag{\ref{L:COMPLETENESS}.1}\]

    Therefore, the argument splits into three cases below.\\

\noindent Case 1. $t_1 \equiv \bot$.

        Then, by $DI1$, $DI4$, $DI5$ and TRANS, we have $\vdash  t_1 \leqslant t_2$.\\

\noindent Case 2. $t_1 \equiv 0$.

       Clearly, $t_1 \notin F $ and ${\mathcal I}(t_1)=\emptyset$. Further we get $t_2 \notin F $ and ${\mathcal I}(t_1)={\mathcal I}(t_2)$ by $t_1 \underset{\thicksim}{\sqsubset}_{RS} t_2$. Thus, by (\ref{L:COMPLETENESS}.1), we have $t_2 \equiv 0$. Then $\vdash  t_1 \leqslant t_2$ follows from REF.\\

\noindent Case 3. $t_1 \equiv \underset{i< n}{\square}a_i.t_{1i} $ with $n>0$.

        Since $t_1 \in NF_B\subseteq T(\Sigma_B)$, by Lemma~\ref{L:BPT}, we have $t_1 \notin F $.
        Hence, by $t_1 \underset{\thicksim}{\sqsubset}_{RS} t_2$,  we get $t_2 \notin F $ and ${\mathcal I}(t_2)={\mathcal I}(t_1)=\{a_i|i < n\} \not= \emptyset$.
        Further, it follows from (\ref{L:COMPLETENESS}.1) and the condition (D) in Def.~\ref{D:NORMAL_FORM} that
        there exist $t_{2i}\in NF_B$ and $a_i'\in Act$($i<n$) such that
         \[t_2 \equiv \underset{i < n}{\square}a_i'.t_{2i}\in NF_B\;\text{and}\; \{a_i|i<n\}=\{a_i'|i<n\}.\]
       By CONTEXT, it is easy to know that, in order to complete the proof, it is sufficient to show that
       \[\forall i<n \exists i'<n(\vdash a_i.t_{1i}\leqslant a_{i'}'.t_{2i'}).\]
    Let $i_0< n$. We have $a_{i_0}=a_{i_0'}'$ for some $i_0'<n$.
    Since $t_{1i_0},t_{2i_0'}\in NF_B$, by Def.~\ref{D:NORMAL_FORM}, there exist $m,m'>0$, $s_j(j<m)$ and $s_{j'}'(j'<m')$ such that
    \begin{enumerate}
      \item $t_{1i_0} \equiv  \underset{j< m}{\bigvee}s_{j}$ and $t_{2i_0'}\equiv  \underset{j'< m'}{\bigvee} s_{j'}'$,
      \item $s_j$ and $s_{j'}'$ are stable for each $j<m$ and $j'<m'$,
      \item $s_j,s_{j'}' \in NF_B$ for each $j<m$ and $j'<m'$.
    \end{enumerate}
    In the following, we want to show that $\vdash  s_{j}\leqslant t_{2i_0'}$ for each $j< m$.
     Let $j_0< m$.
     Since $NF_B \subseteq T(\Sigma_B)$, by Lemma~\ref{L:BPT} and \ref{L:SUBSTRUCTURE}(1), it immediately follows that $t_{1i_0}\stackrel{\epsilon}{\Rightarrow}_F|s_{j_0}$.
    Thus, $t_1 \stackrel{a_{i_0}}{\rightarrow}_F t_{1i_0} \stackrel{\epsilon}{\Rightarrow}_F| s_{j_0}$.
    Then, it follows from  $t_1 \underset{\thicksim}{\sqsubset}_{RS} t_2$ that
    \[t_2 \stackrel{a_{i_0}}{\Rightarrow}_F|t_2' \;\text{and}\; s_{j_0}\underset{\thicksim}{\sqsubset}_{RS} t_2'\;\text{for some}\; t_2'.\tag{\ref{L:COMPLETENESS}.2}\]
    Further, since $t_2$ is injective in prefixes and $t_2$ is stable, we get $t_2 \stackrel{a_{i_0}}{\rightarrow}_F t_{2i_0'}\stackrel{\epsilon}{\Rightarrow}_F|t_2'$.
    Then, by Lemma~\ref{L:SUBSTRUCTURE}(2), we obtain
    \[t_2'\equiv s_{j_0'}' \;\text{for some}\; j_0'< m'.\tag{\ref{L:COMPLETENESS}.3}\]
    Since $|t_1|>|s_{j_0}|$, by (\ref{L:COMPLETENESS}.2), (\ref{L:COMPLETENESS}.3) and IH, we get $\vdash  s_{j_0}\leqslant s_{j_0'}'$.
    Further, by $DI1$, $DI2$, $DI5$ and TRANS, we have $\vdash  s_{j_0}\leqslant t_{2i_0'}$, as desired.

    So far, we have obtained
    \[\vdash s_{j} \leqslant t_{2i_0'}\;\text{for each}\;j<m.\]
    Then, by $DI1$, $DI2$, $DI3$, CONTEXT and TRANS, we get $\vdash \underset{j<m}{\bigvee}s_{j} \leqslant t_{2i_0'}$, that is, $\vdash t_{1i_0}\leqslant t_{2i_0'}$.
    So, by CONTEXT, it follows that $\vdash  a_{i_0}.t_{1i_0} \leqslant a_{i_0'}'.t_{2i_0'}$.
\end{proof}

We are now ready to prove the main result of this section.

\begin{theorem}[Ground-Completeness]
  For any $t_1,t_2 \in T(\Sigma_{\text{CLL}})$, $t_1 \sqsubseteq_{RS} t_2$ implies $ \vdash t_1 \leqslant t_2$.
\end{theorem}
\begin{proof}
  Assume that $t_1 \sqsubseteq_{RS} t_2$.
  By Theorem~\ref{T:NORMALFORM}, $\vdash  t_1 = t_1^*$ and $\vdash  t_2 = t_2^*$ for some $t_1^*,t_2^* \in NF$.
  It suffices to prove that $\vdash  t_1^* \leqslant t_2^*$.
  By Theorem~\ref{T:SOUNDNESS}, we have $t_1 =_{RS} t_1^*$ and $t_2 =_{RS} t_2^*$.
  So $t_1^* \sqsubseteq_{RS} t_2^*$.

   If $t_1^* \equiv \bot$ then it follows from $DI1$, $DI4$, $DI5$ and TRANS that $\vdash  t_1^* \leqslant t_2^*$.
   Next, we consider the case $t_1^* \not\equiv \bot$. Then, $t_1^* \in NF_B$.
   We may assume $t_1^* \equiv \underset{i< n}{\bigvee}t_{1i}$ with $n>0$ and for each $i<n$, $t_{1i} \equiv \underset{j<m_i}\square a_{ij}.r_{ij}\in NF_B$ with $m_i \geq 0$.
    In order to complete the proof, it is sufficient to show that
    \[\vdash t_{1i}\leqslant t_2^*\;\text{for each}\;i< n.\]
    Let $i_0 < n$.
    Since $NF_B \subseteq T(\Sigma_B)$, by Lemma~\ref{L:BPT} and \ref{L:SUBSTRUCTURE}(1), we have $t_1^* \stackrel{\epsilon}{\Rightarrow}_F|t_{1i_0}$.
    Then, it follows from $t_1^* \sqsubseteq_{RS} t_2^*$ that $t_2^* \stackrel{\epsilon}{\Rightarrow}_F| t_2'$ and $t_{1i_0}\underset{\thicksim}{\sqsubset}_{RS} t_2'$ for some $t_2'$.
    So, $t_2^* \notin F $, that is, $t_2^* \not\equiv \bot$.
    Thus, $t_2^* \in NF_B$ and we may assume that  $t_2^* \equiv \underset{i< k}{\bigvee}t_{2i}$ with $k>0$ and for each $i<k$, $t_{2i} \equiv \underset{j<m_i'}\square b_{ij}.s_{ij}\in NF_B$ for some $m_i' \geq 0$.
    Thus, $t_{2i}$ is stable for each $i<k$.
    Then, by Lemma~\ref{L:SUBSTRUCTURE}(2), it follows from $t_2^* \stackrel{\epsilon}{\Rightarrow}_F| t_2'$ that $t_2' \equiv t_{2i_0'}$ for some $i_0'< k$.
    Further, by Lemma~\ref{L:COMPLETENESS}, $\vdash t_{1i_0}\leqslant t_{2i_0'}$ follows from $t_{1i_0}\underset{\thicksim}{\sqsubset}_{RS} t_2' \equiv t_{2i_0'}$.
    Finally, by $DI1$, $DI2$, $DI5$ and TRANS, we obtain $\vdash t_{1i_0}\leqslant t_2^*$, as desired.
\end{proof}

\section{Conclusions and Discussion}

  This paper has  provided a ground-complete proof system for weak ready simulation presented by L{\"u}ttgen and Vogler for the finite fragment of the calculus $\text{CLL}_R$.
 In addition to standard axioms, since enriching process languages with logical operators conjunction and disjunction, such proof system contains a number of axioms to capture the interaction between usual process operators and logical operators.

    Compared with usual notions of behaviour preorders \cite{Glabbeek01}, a specific point of L\"{u}ttgen and Vogler's ready simulation is that it involves consideration of inconsistencies.
    The predicate $F$ plays a central role in this notion.
    Due to such particular characteristic, side-conditions are attached to some axioms in $AX_{\text{CLL}}$ (including $DS4$, $ECC3$ and $EXP2$) so that processes can be treated differently according to their consistency.
%    the system $AX_{\text{CLL}}$ contains a few of axioms to handle inconsistent processes.
    The guideline in designing of $AX_{\text{CLL}}$ is that we need to find enough axioms to reduce (in)consistent processes to basic processes ($\bot$, resp.).
    Such trick seems to be also useful in considering proof system for more general cases involving recursions.
    However, it is far from trivial to carry out this trick in the presence of recursions.
    In the following, we would like to discuss this sketchily.

In the framework of LLTS, since divergence is viewed as catastrophic, any process, which cannot evolve into a stable state in finitely many steps, is specified to be inconsistent.
This intuition is captured formally by the condition (LTS2) in Def.~\ref{D:LLTS}.
Obviously, it is recursion that may bring divergence.
Thus we must put attention to such additional origin of inconsistency in the presence of recursions.

In order to carry out the trick mentioned above,
%for the full calculus (more precisely, the fragment consisting  of regular process terms \footnote{A regular process is a process with finite consistent state.}),
we need to isolate a particular subclass of terms syntactically, which plays a role analogous to that played by $T(\Sigma_B)$ (see Def.~\ref{D:BPT}) in this paper.
In our mind, a rational choice for such subclass is $ET(\Sigma_B)$ mentioned in Remark~\ref{R:EX_BPT}, which extends $T(\Sigma_B)$ by admitting strongly guarded processes $\langle X|E \rangle$ (without involving conjunction and $\bot$) into BNF grammar of $T(\Sigma_B)$, and satisfies $ET(\Sigma_B) \cap F = \emptyset$ (its proof is given in the Appendix).
%Sofar the subclass $ET(\Sigma_B)$ mentioned in Remark~\ref{R:EX_BPT} seems to be the one that we desire.
%First  $ET(\Sigma_B) \cap F = \emptyset$; second, this subclass extends $T(\Sigma_B)$ by admitting strongly guarded processes $\langle X|E \rangle$ (without involving conjunction and $\bot$) into BNF grammar of $T(\Sigma_B)$, this extension is straightforward without considering many other consistent terms with unknown formats.

To confirm that the choice above is right, we must ensure that $ET(\Sigma_B)$ is sufficiently expressive to ``represent'' all consistent processes.
That is, we need to provide a group of axioms so that, for any process $t$, if $t$ is (in)consistent then it can be reduced to one in $ET(\Sigma_B)$ ($\bot$ resp.) by applying these axioms.
At present, it seems to be difficult to find these axioms.
For instance, since there exist weakly guarded recursions that is consistent (e.g., $\langle X|X= (X \Box a.0) \vee b.0 \rangle$), we need enough axioms to transfer them into $ET(\Sigma_B)$.
In particular, a few axioms are needed to transfer (consistent) weakly guarded recursions into strongly guarded ones (notice that all recursive processes in $ET(\Sigma_B)$ are strongly guarded).
In \cite{Milner89b}, Milner has solved analogous problem for observational congruence in the calculus CCS through referring the following axioms \footnote{In \cite{Milner89b}, Milner uses the operator $+$ and the notation $\mu X t$ instead of external choice $\Box$ and $\langle X|X=t \rangle$ resp. Moreover Baeten and Bravetti point out that Axioms (M2) and (M3) can be equivalently expressed by a single axiom \cite{Baeten08}.}.
\[\langle X | X = X \Box t \rangle =\langle X| X= t\rangle \tag{M1}\]
  \[\langle X|X=\tau.X \Box t \rangle = \langle X|X= \tau.t\rangle\tag{M2}\]
  \[\langle X| X=\tau.(X \Box t) \Box s \rangle = \langle X|X=\tau.X \Box t \Box s\rangle\tag{M3}\]
  Unfortunately, none of these axioms works well in our situation.
  First, since  unguarded recursions are incompatible with negative rules \cite{Bloom94}, the calculus $\text{CLL}_R$ restricts itself to guarded ones \cite{Zhang14}.
  Hence Axiom (M1) is outside our terms of reference.
  Second, Axiom (M2) is not valid w.r.t $=_{RS}$.
  For instance, consider $t\equiv a.X$, then we get $\langle X|X=\tau.X \Box a.X \rangle \in F$ and $\langle X|X=\tau.a.X \rangle \notin F$.
 Finally, due to $\tau$-purity,  both $\langle X|X=\tau.(X\Box t)\Box s\rangle$ and $\langle X|X=\tau.X \Box t \Box s \rangle$ are inconsistent for any $t,s$.
 Therefore, Axiom (M3) may be useful for transferring inconsistent processes into $\bot$ because the scope of the prefix $\tau.()$ in left-hand side of (M3) is larger than one in right-hand side, but it no longer has any effect on transferring consistent weakly guarded $\langle X|E \rangle$ into strongly guarded one.

 Summarily, we need to find appropriate axioms from scratch to cope with inconsistency caused by recursions.

\appendix
\section{Appendix}

We mentioned in Section~5 that $ET(\Sigma_B) \cap F = \emptyset$.
This Appendix is devoted to proving this claim.
We first define $ET(\Sigma_B)$ formally.

 \begin{mydefn}[Extended Basic Term]
   The extended basic terms are defined by BNF: $t::=0\mid (\alpha.t) \mid t\Box t \mid t \vee t \mid t \parallel_A t \mid X \mid \langle X|X=t\rangle$, where $\alpha \in Act_{\tau}$, $X \in V_{AR}$, $A \subseteq Act$ and in $\langle X|X=t\rangle$, $X$ is strongly guarded in $t$.
   We denote $ET(\Sigma_B)$ as the set of all extended basic terms.
 \end{mydefn}

As usual, we use $t_{\widetilde{X}}$ to denote a term $t$ whose free variables form a subset of $\{X_1,\dots,X_n\}$ where $\widetilde{X}=(X_1,\dots,X_n)$ is a $n$-tuple distinct variables.
$t_{\widetilde{X}}$ is stable if $t_{\widetilde{X}}\{\widetilde{\tau.0}/\widetilde{X}\}\not\stackrel{\tau}{\rightarrow}$.

\begin{lemma}\label{L:Stable}
  If $t_{\widetilde{X}}$ is stable then $t_{\widetilde{X}}\{\widetilde{p}/\widetilde{X}\} \not\stackrel{\tau}{\rightarrow}$ for any $\widetilde{p}$.
\end{lemma}
\begin{proof}
  Assume $t_{\widetilde{X}}\{\widetilde{p}/\widetilde{X}\} \stackrel{\tau}{\rightarrow} r$ for some $r$.
  It suffices to prove $t_{\widetilde{X}}\{\widetilde{\tau.0}/\widetilde{X}\} \stackrel{\tau}{\rightarrow}$.
  It proceeds by induction on the depth of the inference of $Strip(\text{CLL}_R,M_{\text{CLL}_R}) \vdash t_{\widetilde{X}}\{\widetilde{p}/\widetilde{X}\} \stackrel{\tau}{\rightarrow} r$.
The induction is easy to carry out by distinguishing several cases based on the last rule applied in the inference.
We leave the proof to the reader.
\end{proof}

 %In the following, we intend to prove that for any process (i.e. term without free variables) $p \in ET(\Sigma_B)$, $p \stackrel{\epsilon}{\Rightarrow}| p'$ for some process $p' \in ET(\Sigma_B)$. Formally, it suffice to prove:

 \begin{lemma}\label{L:EBPB}
   If $t_{\widetilde{X}}$ %\footnote{$\widetilde{X} = (X_1,\dots,X_n)(n \geq 0)$ is a $n$-tuple distinct variables. Given $\widetilde{p}=(p_1,\dots,p_n)$, $t_{\widetilde{X}}\{\widetilde{p}/\widetilde{X}\} \equiv t_{\widetilde{X}}\{p_1/X_1,\dots,p_n/X_n\}$.}
   is a term in  $ET(\Sigma_B)$ such that $X$ is strongly guarded in $t_{\widetilde{X}}$  for each $X \in \widetilde{X}$, then there exists $t_{\widetilde{X}}' \in ET(\Sigma_B)$ such that, for any $\widetilde{q}$,
   %$t_{\widetilde{X}}\{\widetilde{0}/\widetilde{X}\} \stackrel{\epsilon}{\Rightarrow}|t_{\widetilde{X}}'\{\widetilde{0}/\widetilde{X}\}$ for some $t_{\widetilde{X}}'$  such that
   $t_{\widetilde{X}}\{\widetilde{q}/\widetilde{X}\} \stackrel{\epsilon}{\Rightarrow}|t_{\widetilde{X}}'\{\widetilde{q}/\widetilde{X}\}$.
   %Moreover if $t_{\widetilde{X}}$ does not contain any free variables then neither does $t_{\widetilde{X}}'$.
 \end{lemma}
 \begin{proof}
   If $t_{\widetilde{X}}$ is stable then the conclusion holds trivially by Lemma~\ref{L:Stable}.
   In the following, we devote ourselves to considering non-trivial case where $t_{\widetilde{X}}$ is not stable.
   It proceeds by induction on the structure of $t_{\widetilde{X}}$.
   Here we consider only non-trivial case
   %It is a routine case analysis based on the structure of $t$.
%
%\noindent Case 1. $t_{\widetilde{X}} \equiv \tau.s_{\widetilde{X}}$.
%
%It is trivial. \\
%
%\noindent Case 2. $t_{\widetilde{X}} \equiv s_{\widetilde{X}} \odot r_{\widetilde{X}}$ with $\odot \in \{\vee,\Box,\parallel_A\}$.
%
% We handle $t_{\widetilde{X}} \equiv s_{\widetilde{X}} \Box r_{\widetilde{X}}$, others are treated similarly.
% By IH, there exist $s_{\widetilde{X}}',r_{\widetilde{X}}' \in ET(\Sigma_B)$ such that
% $s_{\widetilde{X}}\{\widetilde{q}/\widetilde{X}\} \stackrel{\epsilon}{\Rightarrow}|s_{\widetilde{X}}'\{\widetilde{q}/\widetilde{X}\}$ and $r_{\widetilde{X}}\{\widetilde{q}/\widetilde{X}\} \stackrel{\epsilon}{\Rightarrow}|r_{\widetilde{X}}'\{\widetilde{q}/\widetilde{X}\}$  for any $\widetilde{q}$.
% Hence $t_{\widetilde{X}}\{\widetilde{q}/\widetilde{X}\} \stackrel{\epsilon}{\Rightarrow}|(s_{\widetilde{X}}' \Box r_{\widetilde{X}}') \{\widetilde{q}/\widetilde{X}\}$ for any $\widetilde{q}$. \\
%\noindent Case 3.
$t_{\widetilde{X}} \equiv \langle Y |Y =t \rangle$. In this situation, $Y \notin \widetilde{X}$ and $t$ is in $ET(\Sigma_B)$ whose free variables are in $\{Y\} \cup \widetilde{X}$.
   Moreover, for each $Z \in \{Y\} \cup \widetilde{X}$, $Z$ is strongly guarded in $t$.
   Hence, by IH, there exists %$t_{\widetilde{X},Y}'$ such that   $t_{\widetilde{X},Y}\{\widetilde{p}/\widetilde{X},q/Y\} \stackrel{\epsilon}{\Rightarrow}|t_{\widetilde{X},Y}'\{\widetilde{p}/\widetilde{X},q/Y\}$ for any $\widetilde{p},q$.
   $t' \in ET(\Sigma_B)$ such that $t\{\widetilde{p}/\widetilde{X},q/Y\} \stackrel{\epsilon}{\Rightarrow}|t'\{\widetilde{p}/\widetilde{X},q/Y\}$ for any $\widetilde{p},q$.
   In particular, we get
   \[t\{\widetilde{p}/\widetilde{X},\langle Y |Y=t \rangle \{\widetilde{p}/\widetilde{X}\}/Y\} \stackrel{\epsilon}{\Rightarrow}|t'\{\widetilde{p}/\widetilde{X},\langle Y |Y=t \rangle \{\widetilde{p}/\widetilde{X}\}/Y\}\;\text{for any}\;\widetilde{p}.\]
   Further, by Rule $Ra_{16}$, it follows from $t\{\widetilde{p}/\widetilde{X},\langle Y |Y=t \rangle \{\widetilde{p}/\widetilde{X}\}/Y\} \equiv t\{\langle Y|Y=t\rangle /Y\}\{\widetilde{p}/\widetilde{X}\}$ that
   \[\langle Y|Y=t \rangle \{\widetilde{p}/\widetilde{X}\} \stackrel{\epsilon}{\Rightarrow}|t'\{\langle Y |Y=t \rangle /Y\}\{\widetilde{p}/\widetilde{X}\}\;\text{for any}\;\widetilde{p}.\]
%   where
%   \[t'' \equiv \left\{
%                  \begin{array}{ll}
%                    Y, & t\;\hbox{is stable;} \\
%                    t', & \hbox{otherwise.}
%                  \end{array}
%                \right.\]
   Set $t_{\widetilde{X}}''\triangleq t'\{\langle Y|Y=t \rangle /Y\}$.
   Then it is easy to see that $t_{\widetilde{X}}'' \in ET(\Sigma_B)$ due to $t',\langle Y|Y=t \rangle \in ET(\Sigma_B)$.
   Hence $t_{\widetilde{X}}''$ is the one that we desire.
 \end{proof}

As an immediate consequence of the lemma above, we have
\begin{corollary}\label{C:EBPT}
  For any process (i.e., terms with no free variables) $p \in ET(\Sigma_B)$, there exists $q\in ET(\Sigma_B)$ such that $p \stackrel{\epsilon}{\Rightarrow}|q$.
\end{corollary}

\begin{proposition}
  $ET(\Sigma_B) \cap F =\emptyset$.
\end{proposition}
\begin{proof}
  Since $F$ is a set of processes, it suffices to show that each process in $ET(\Sigma_B)$ is consistent.
  Let $\Omega$ be the set of all processes in $ET(\Sigma_B)$.
  Due to the well-foundedness of proof trees, in order to complete the proof, it is sufficient to show that, for any $p \in \Omega$, if  $\mathcal T$ is a proof tree of $Strip({\text{CLL}_R,M_{\text{CLL}_R}})\vdash pF$ then $\mathcal T$ has a proper subtree with root $rF$ for some $r \in \Omega$.
  We shall prove this as follows.

  Let $p \in \Omega$ and $\mathcal T$ be a proof tree of $pF$.
  It is a routine case analysis based on the last rule applied in $\mathcal T$. %the proof tree of $pF$.
  We distinguish different cases based on the form of $p$.
  Clearly,  $p \not\equiv 0$ due to $0 \notin F$.
  For $p \equiv \alpha.p_1$ or $p_1 \odot p_2$ with $\odot \in \{\vee, \Box, \parallel_A\}$, it is obvious that $p_1,p_2 \in ET(\Sigma_B)$.
  Moreover, by SOS rules of $\text{CLL}_R$, it is easy to see that $\mathcal T$ has a proper subtree with root $p_iF$ for some $i \in \{1,2\}$.
  Next we handle the case $p \equiv \langle Y|Y=t_Y \rangle$.
  Then the last rule applied in $\mathcal T$ is either $\frac{\langle t_Y|Y=t_Y\rangle F}{\langle Y|Y=t_Y\rangle F}$ or $\frac{\{rF:\langle Y|Y=t_Y\rangle  \stackrel{\epsilon}{\Rightarrow}|r\}}{\langle Y|Y=t_Y\rangle F}$.
  For the former, it is obvious that $\langle t_Y|Y=t_Y\rangle \in \Omega$ due to $\langle t_Y |Y =t_Y \rangle \equiv t_Y \{\langle Y|Y=t_Y \rangle /Y\}$ (see subsection~2.2) and $t_Y,\langle Y|Y=t_Y \rangle \in ET(\Sigma_B)$.
  For the latter, by Corollary~\ref{C:EBPT}, $\langle Y|Y=t_Y\rangle  \stackrel{\epsilon}{\Rightarrow}|r'$ for some $r' \in \Omega$, as desired.
\end{proof}

\end{document}